\documentclass[12pt]{article}
\usepackage{enumerate}
\usepackage{amsfonts}
\usepackage{latexsym}
\usepackage{color}
\usepackage{graphicx}
\usepackage{wrapfig}
\usepackage{caption}
\usepackage{float}

\usepackage{amsmath,amssymb,amsthm}
\usepackage{bm}
\usepackage{braket}

\def\D{\mathcal{D}}
\def\HH{\mathcal{H}}

\def\St{\mathfrak{S}}
\def\P{\mathfrak{P}}

\def\C{\mathfrak{C}}
\def\T{\mathfrak{T}}
\def\B{\mathfrak{B}}

\newcommand{\supp}{\mathrm{supp}}
\newcommand{\rank}{\mathrm{rank}}

\newcommand{\id}{\mathrm{Id}}
\newcommand{\Tr}{\mathrm{Tr}}
\newcommand{\ext}{\mathrm{ext}}

\newcommand{\shs}{\hspace{1pt}}
\newcounter{defin}  \newcounter{lemma}  \newcounter{theorem}
\newcounter{proposition} \newcounter{corol}  \newcounter{remark} \newcounter{example}
\newenvironment{lemma}{\par\refstepcounter{lemma}     \textbf{Lemma \thelemma.} }{\rm\par}
\newenvironment{theorem}{\par\refstepcounter{theorem}     \textbf{Theorem \thetheorem.}\ }{\rm\par}
\newenvironment{proposition}{\par\refstepcounter{proposition}     \textbf{Proposition \theproposition.}\ }{\rm\par}
\newenvironment{corollary}{\par\refstepcounter{corol}     \textbf{Corollary \thecorol.} }{\rm\par}

\newenvironment{remark}{\par\refstepcounter{remark}     \textbf{Remark \theremark.}}{\rm\par}
\newenvironment{example}{\par\refstepcounter{example}     \textbf{Example \theexample.}}{\rm\par}

\begin{document}

\title{The $\varepsilon$-smoothed Entanglement of Formation and its relation to the Moreau-Yosida approximation}




\author{M.E.~Shirokov\footnote{email:msh@mi.ras.ru}\\
Steklov Mathematical Institute, Moscow, Russia}
\date{}
\maketitle

\begin{abstract}
 We describe a family of convex continuous functions $E_{F,\varepsilon}$, $\varepsilon\in(0,1]$, on the set of states of a bipartite quantum system (consisting of finite-dimensional or infinite-dimensional subsystems) obtained by applying the  Mora-Piani-Briegel construction of  $\varepsilon$-measure (proposed in \cite{MPB}) to the Entanglement of Formation (EoF) with the use of the trace norm distance in the role of basic metric. These functions are non-increasing in \(\varepsilon\) and converge pointwise  to the EoF as $\varepsilon\to0^+$.  We conducted a detailed analysis of  properties of the functions $E_{F,\varepsilon}$, using (among others) the relationship between the family $\{E_{F,\varepsilon}\}_{\varepsilon\in[0,1]}$ and the family $\{E^\lambda_F(\rho)\}_{\lambda>0}$ of Moreau-Yosida approximations of the EoF introduced and explored in \cite{EF}.

We obtain a variational expression for the functions $E_{F,\varepsilon}$  (which for $\varepsilon=0$  turns into the standard variational expression for the EoF).
We also describe the weak subadditivity property of the functions $E_{F,\varepsilon}$ (after proving  that the usual subadditivity does not hold for these functions) and analyse their continuity properties.

The semicontinuity bounds for the EoF (presented in \cite{LCB+,AOE}) allow us to obtain easily computable upper bounds on the difference $E_F(\rho)-E_{F,\varepsilon}(\rho)$ for a given state $\rho$.  These bounds give easily computable bounds on the rate of uniform convergence of the function $E_{F,\varepsilon}$ to the EoF  as $\varepsilon\to0^+$ on the sets of states with bounded rank/energy of one of the marginal states.

Applicability of the obtained results to other entanglement measures, in particular, to the relative entropy of entanglement, is discussed.

As an auxiliary result, it is shown that for any state of an infinite-dimensional bipartite quantum system, there exists a nearest separable state w.r.t. the trace norm distance. The non-uniqueness of such a state is also shown.
\end{abstract}
\newpage

\tableofcontents

\section{Introduction}

The idea of smoothing quantities over a neighbourhood of an
imperfectly known object has a history predating its application
to entanglement measures. In information-theoretic and
cryptographic settings, entropy-smoothing ideas were used in
the construction of pseudorandom generators
\cite{ILL}. A systematic framework of
smooth entropy was subsequently developed by Cachin and Maurer
\cite{C&M}. Renner and Wolf introduced smooth
R\'enyi entropies defined through optimization over
neighbourhoods of probability distributions
\cite{R&W}. These works provide conceptual and
mathematical precedents for the $\varepsilon$-measure
construction considered here, although they concern entropy
and randomness extraction rather than entanglement measures.\smallskip

Mora, Piani, and Briegel introduced in  \cite{MPB} a general $\varepsilon$-measure
construction which associates with an entanglement measure $E$ and a
distance $D$ the quantity
\begin{equation}
E_{\varepsilon}^{(D)}(\rho)
=
\inf_{\sigma:\,D(\rho,\sigma)\leq\varepsilon} E(\sigma).
\label{eq:MPB-general}
\end{equation}
The construction is intended to quantify how much entanglement is
guaranteed to be present when the state is known only up to a prescribed
precision.  The same\break $\varepsilon$-measure construction has
subsequently been considered in broader quantum resource theories
\cite{LuoMengWang2024}.

In this article, we use the trace norm distance
\begin{equation}
T(\rho,\sigma)
\doteq
\frac12\|\rho-\sigma\|_1
\end{equation}
in the role of metric $D$ and define the corresponding $\varepsilon$-smoothed Entanglement of Formation ($\varepsilon$-smoothed EoF) as follows
\begin{equation}
E_{F,\varepsilon}(\rho)
\doteq
\inf_{\sigma:\,T(\rho,\sigma)\leq\varepsilon}
E_F(\sigma)
\label{eq:smooth-EF}
\end{equation}
for an arbitrary state $\rho$ of a bipartite system and  $\varepsilon\in[0,1]$ ($E_{F,0}(\rho)=E_F(\rho)$).

In finite dimensions the infimum in (\ref{eq:smooth-EF}) is attained.  In infinite dimensions
we keep the infimum notation.

In Section 3,  we describe basic properties of the $\varepsilon$-smoothed EoF in general infinite-dimensional settings
and consider its relation to the Moreau--Yosida approximation of the EoF defined as (cf. \cite{EF})
\begin{equation}\label{EFA}
 E_F^\lambda(\rho)=\inf_{\sigma}\left\{E_F(\sigma)+\frac{1}{\lambda}\,T(\rho,\sigma)\right\},\qquad \lambda>0,
\end{equation}
We prove the variational relation
\[
E_F^\lambda(\rho)
=
\inf_{0\leq\varepsilon\leq1}
\left\{
E_{F,\varepsilon}(\rho)+\frac{\varepsilon}{\lambda}
\right\},\qquad \lambda>0,
\]
and, using convex duality, its converse
\begin{equation}\label{eq:inverse-duality}
E_{F,\varepsilon}(\rho)
=
\sup_{\lambda>0}
\left\{
E_F^\lambda(\rho)-\frac{\varepsilon}{\lambda}
\right\},\qquad \varepsilon\in[0,1].
\end{equation}

By these relations for any given state $\rho$ the families $\{E_F^\lambda(\rho)\}_{\lambda>0}$ and $\{E_{F,\varepsilon}(\rho)\}_{\varepsilon\in[0,1]}$
contain the same information about this state $\rho$. Moreover, they allow us to analyse properties of one of these families by exploring properties of the other.
In this article, we  use the detailed investigation of the Moreau--Yosida approximation of the EoF  made in \cite{EF} to prove new results
for the $\varepsilon$-smoothed EoF (and to simplify the general proofs presented in \cite{MPB} in this special case).\smallskip

In particular, dual representation \eqref{eq:inverse-duality} allows us to obtain a variational expression for the functions $E_{F,\varepsilon}$  (which for $\varepsilon=0$  turns into the standard variational expression for the EoF).\smallskip

After proving  that the usual subadditivity does not hold for the function $E_{F,\varepsilon}$, we discuss the  weak
subadditivity relation: whenever
$\varepsilon_1,\varepsilon_2\geq0$ and
$\varepsilon_1+\varepsilon_2\leq1$, the following inequality holds
\begin{equation}
E_{F,\varepsilon_1+\varepsilon_2}(\rho\otimes\sigma)
\leq
E_{F,\varepsilon_1}(\rho)
+
E_{F,\varepsilon_2}(\sigma),
\label{eq:weak-subadditivity}
\end{equation}
where the tensor product is treated with respect to the
corresponding bipartition of the combined systems. This property can be either proved directly or deduced easily
from the subadditivity of $E_F^\lambda$ (proved in \cite{EF}) and expression
\eqref{eq:inverse-duality}. \smallskip

In Section 4, we  study the relation between fixed-radius smoothing
and regularization. Define
\begin{equation}
E^\infty_{F,\varepsilon}(\rho)
=
\limsup_{n\to\infty}
\frac1n E_{F,\varepsilon}(\rho^{\otimes n}),
\label{eq:regularized-smoothed-eof}
\end{equation}
and, whenever the regularized EoF
\[
E_F^\infty(\sigma)
=
\lim_{n\to\infty}\frac1n E_F(\sigma^{\otimes n})
\]
is defined, introduce the alternative quantity
\begin{equation}
 [E_F^\infty]_\varepsilon(\rho)
 \doteq \inf_{\substack{\sigma:\shs T(\rho,\sigma)\le\varepsilon}}
      E_F^\infty(\sigma).
 \label{eq:smoothed-regularized-eof-intro}
\end{equation}
These constructions perform the smoothing and regularization
operations in different orders. An example showing that they cannot be
identified is presented and some related questions are considered.

Attempts to prove the positivity of $E_{F,\varepsilon}(\rho)$  for states $\rho$ whose $\varepsilon$-vicinity  does not contain separable states (in the infinite-dimensional settings) lead us to
the question of attainability of the infimum in the definition
\[
 d\!\left(\rho,\St^{AB}_{\rm sep}\right)=\frac{1}{2}\inf_{\sigma\in\St_{\rm sep}^{AB}}\|\rho-\sigma\|_1
\]
of the distance from an arbitrary state $\rho$ of a system $AB$ to the set $\St^{AB}_{\rm sep}$ of all separable states
of this  system. This question is
resolved in the positive sense (despite the non-compactness of the set $\St^{AB}_{\rm sep}$) using the weak* closedness of the cone generated by $\St^{AB}_{\rm sep}$ \cite{L-Sh-1}. In other words, it is shown that for any state of an infinite-dimensional bipartite quantum system, there exists a nearest separable state w.r.t. the trace norm distance. This result (which seems useful beyond the scope of this article), its proof and an example showing the non-uniqueness of a nearest separable state are presented in the Appendix.

\section{Preliminaries}

Throughout the article, we assume that $\mathcal{H}_X$ is a separable Hilbert space describing a quantum system $X$,
$\mathfrak{B}_{\rm sa}(\mathcal{H}_X)$ is the real Banach space of all Hermitian bounded operators on $\mathcal{H}_X$ with the operator norm $\|\cdot\|$ and $\mathfrak{T}(\mathcal{H}_X)$ is the
(complex) Banach space of all trace-class
operators on $\mathcal{H}_X$  with the trace norm $\|\!\cdot\!\|_1$. Write
$\mathfrak{S}(\mathcal{H}_X)$ for   the set of quantum states (positive operators
in $\mathfrak{T}(\mathcal{H}_X)$ with unit trace) \cite{H-SCI,Wilde}. Denote the set of pure states (one-rank projectors) in  $\mathfrak{S}(\mathcal{H}_X)$ by  $\ext\mathfrak{S}(\mathcal{H}_X)$.

Write $I_{X}$ for the unit operator on a Hilbert space
$\mathcal{H}_X$ and  $\id_{X}$ for the identity
transformation of the Banach space $\mathfrak{T}(\mathcal{H}_X)$.

The \emph{von Neumann entropy} of a quantum state
$\rho \in \mathfrak{S}(\HH)$ is  defined by the formula
$S(\rho)=\operatorname{Tr}\eta(\rho)$, where  $\eta(x)=-x\ln x$ if $x>0$
and $\eta(0)=0$. It is a concave lower semicontinuous function on the set~$\mathfrak{S}(\HH)$ taking values in~$[0,+\infty]$ \cite{H-SCI,L-2,W}.\smallskip

A finite or
countable collection $\{\rho_{i}\}$ of states
with a probability distribution $\{p_{i}\}$ is conventionally called
\textit{discrete ensemble} and denoted by $\{p_{i},\rho_{i}\}$. The state
$\bar{\rho}\doteq\sum_{i}p_{i}\rho_{i}$ is called \emph{average state} of this  ensemble. \smallskip

A \textit{generalized (continuous) ensemble} is defined as
a Borel probability measure on the set of quantum states \cite{H-SCI,EM}. We denote by $\P(\mathcal{H})$ the set of all Borel probability measures on $\ext\mathfrak{S}(\mathcal{H})$ equipped with the topology of weak convergence
\cite{Bog,Par}.\footnote{The weak convergence of a sequence $\{\mu_n\}\subset\P(\mathcal{H})$  to a measure $\mu_0\in\P(\mathcal{H})$ means that
$\,\lim_{n\rightarrow\infty}\int f(\rho)\mu_n(d\rho)=\int f(\rho)\mu_0(d\rho)\,$
for any continuous bounded function $f$ on $\,\St(\HH)$.}
 The set $\P(\mathcal{H})$ can be treated as a complete
separable metric space containing the dense subset $\P_{\rm d}(\mathcal{H})$ of discrete measures (corresponding to discrete ensembles) \cite{Par}. The average state of a generalized
ensemble $\mu \in \P(\mathcal{H})$ is the barycenter of the measure
$\mu $ defined by the Bochner integral
\begin{equation}\label{bd}
\bar{\rho}(\mu )=\int_{\ext\mathfrak{S}(\mathcal{H})}\rho \mu (d\rho ).
\end{equation}

The Entanglement of Formation (EoF) is one of the basic entanglement measures in bipartite quantum systems \cite{P&V,4H}. For quantum systems  $A$ and $B$  (of any finite or infinite dimensions)  the EoF of a state $\rho$ of the composite system $AB$ is defined as the maximal \emph{closed} (lower semicontinuous) convex extension of the function
$\rho\mapsto S(\rho_A)=S(\rho_B)$ on the set $\ext\St(\HH_{AB})$ of pure states of the system $AB$ to the set $\St(\HH_{AB})$ of all states of this system, which can be defined as
\begin{equation}\label{EF-d}
E_F(\rho)=\inf_{\sum_k\!p_k\rho_k=\rho}\sum_kp_kS([\rho_k]_A)=\inf_{\int\varrho\mu(d\varrho)=\rho}\int_{\ext\St(\HH_{AB})}S(\varrho_A)\mu(d\varrho),
\end{equation}
where the first infimum is over all discrete (finite or countable) ensembles $\{p_k, \rho_k\}$ of pure states in $\St(\HH_{AB})$ with the average state $\rho$ and
the second infimum is over all  measures in $\P(\mathcal{H}_{AB})$ with the barycenter $\rho$ \cite{Bennett,4H,EoF}.\footnote{Until recently, two versions of the EoF had to be considered in the infinite-dimensional settings, defined in terms of discrete and continuous convex roof extensions. In \cite{EoF} this ambiguity has been eliminated by proving the coincidence of these two extensions.}\smallskip

\begin{remark}\label{imp} (generalization of Theorem 6 in \cite{V&P}) The second expression in (\ref{EF-d}) allows us to prove good analytical properties of the  EoF  (global lower semicontinuity, monotonicity under generalized LOCC, etc. \cite{EM,Lami-new}), while the first one   shows that
\begin{equation}\label{E-ineq}
E(\rho)\leq E_F(\rho)\qquad \rho\in\St(\HH_{AB})
\end{equation}
for any $\sigma$-convex function $E$ on $\St(\HH_{AB})$ (in terms of \cite{EM})\footnote{It means that the function $E$ satisfies discrete Jensen inequality $\,E(\sum_ip_i\sigma_i)\leq \sum_ip_iE(\sigma_i)$.} such that $E(\sigma)=S(\sigma_A)$ for any state $\sigma$ in $\ext\St(\HH_{AB})$ \emph{not assuming
global lower semicontinuity of $E$}.

In quantum information theory, at least three entanglement measures are used for which it is not difficult to prove the $\sigma$-convexity on the set of states of
a bipartite infinite-dimensional quantum system, but the global lower semicontinuity has not yet been proven.
These are the squashed entanglement  $E_{sq}$,
the $c$-squashed entanglement $E_{c\textrm{-}sq}$ and the conditional entanglement of mutual-information $E_{I}$  \cite{Tucci,C&W,Y&C,CE-II}.\footnote{The definitions of these entanglement measures valid for all states of  infinite-dimensional bipartite quantum system can be found in \cite[Examples 4 and 6, Section IV-C]{AM+S}.} So, we may now assert that inequality (\ref{E-ineq}) holds with $E=E_{sq},E_{c\textrm{-}sq},E_I$ in general infinite-dimensional settings.
\end{remark}\medskip

For our purposes it is essential that the EoF defined in (\ref{EF-d}) is a convex lower  semicontinuous function on the set $\St(\HH_{AB})$ and that the second infimum in (\ref{EF-d}) is always attained \cite[Section 5]{EM}. Note that the first infimum in (\ref{EF-d}) is not attained if $\rho$ is a countably-non-decomposable separable state \cite{EoF}.

By introducing the affine lower semicontinuous function
\begin{equation}\label{SA-def}
  S_A(\mu)\doteq\int_{\ext\St(\HH_{AB})}S(\rho_A)\mu(d\rho).
\end{equation}
on the set $\P(\mathcal{H}_{AB})$ the definition (\ref{EF-d}) can be rewritten as
\begin{equation}\label{EF-d+}
  E_F(\rho)=\inf_{\mu\in\P^{\shs\rm d}_\rho(\mathcal{H}_{AB})}S_A(\mu)=\inf_{\mu\in\P_\rho(\mathcal{H}_{AB})}S_A(\mu)=\min_{\mu\in\P_\rho(\mathcal{H}_{AB})}S_A(\mu),
\end{equation}
where $\P_\rho(\mathcal{H}_{AB})$ is the subset of $\P(\HH_{AB})$ consisting of measures $\mu$ such that $\,\bar{\rho}(\mu)=\rho$
and $\P^{\shs\rm d}_\rho(\mathcal{H}_{AB})$ is  the subset of $\P_\rho(\HH_{AB})$ consisting of discrete (purely atomic) measures.
\medskip

Let $H$ be a positive (semi-definite)  operator on a Hilbert space $\mathcal{H}$ (we will always assume that positive operators are self-adjoint). Write  $\mathcal{D}(H)$ for the domain of $H$. For any positive operator $\rho\in\T(\HH)$ we will define the quantity $\Tr H\rho$ by the rule
\begin{equation}\label{H-fun}
\Tr H\rho=
\left\{\begin{array}{ll}
        \sup_n \Tr P_n H\rho\;\; &\textrm{if}\;\;  \supp\rho\subseteq {\rm cl}(\mathcal{D}(H))\\
        +\infty\;\;&\textrm{otherwise}
        \end{array}\right.
\end{equation}
where $P_n$ is the spectral projector of $H$ corresponding to the interval $[0,n]$ and ${\rm cl}(\mathcal{D}(H))$ is the closure of $\mathcal{D}(H)$. If
$H$ is the Hamiltonian (energy observable) of a quantum system described by the space $\HH$ then
$\Tr H\rho$ is the mean energy of a state $\rho$.

For any positive operator $H$ the set
$$
\C_{H,E}=\left\{\rho\in\St(\HH)\,|\,\Tr H\rho\leq E\right\}
$$
is convex and closed (since the function $\rho\mapsto\Tr H\rho$ is affine and lower semicontinuous). It is nonempty if $E> E_0$, where $E_0$ is the infimum of the spectrum of $H$.

In infinite dimensions, the von Neumann entropy is continuous on the set $\C_{H,E}$ for any $E> E_0$ if and only if the operator $H$ satisfies  the \emph{Gibbs condition}
\begin{equation}\label{H-cond}
  \Tr\, e^{-\beta H}<+\infty\quad\textrm{for all}\;\,\beta>0
\end{equation}
and the supremum of the entropy on this set is attained at the \emph{Gibbs state}
\begin{equation*}
\gamma_H(E)\doteq \frac{e^{-\beta_H(E) H}}{\Tr e^{-\beta_H(E) H}},
\end{equation*}
where the parameter $\beta_H(E)$ is determined by the equation $\Tr H e^{-\beta H}=E\Tr e^{-\beta H}$,
and it is assumed that $\gamma_H(E)|\varphi\rangle=0$ for any $\varphi\in \D(H)^{\perp}$ \cite{W}.

Condition (\ref{H-cond}) can be valid only if $H$ is an unbounded operator having  discrete spectrum of finite multiplicity. It means that
this operator has the spectral representation
  \begin{equation}\label{H-form}
H=\sum_{k=1}^{+\infty} h_k |\tau_k\rangle\langle\tau_k|,
\end{equation}
where $\mathcal{T}\doteq\left\{\tau_k\right\}_{k=1}^{+\infty}$ is the orthonormal
system of eigenvectors of $H$ corresponding to the \emph{nondecreasing} unbounded sequence $\left\{h_k\right\}_{k=1}^{+\infty}$ of its eigenvalues. In this case
$\D(H)=\{ \varphi\in\HH_\mathcal{T}\,| \sum_{k=1}^{+\infty} h^2_k |\langle\tau_k|\varphi\rangle|^2<+\infty\}$, where $\HH_\mathcal{T}$ is the closed linear span of $\mathcal{T}$.

For any positive operator $H$ satisfying condition (\ref{H-cond}) we will use the function
\begin{equation}\label{F-def}
F_{H}(E)\doteq\sup_{\rho\in\C_{H,E}}S(\rho)=S(\gamma_H(E))=\beta_H(E)E+\ln Z_H(E)
\end{equation}
on $[h_1,+\infty)$, where $Z_H(E)=\Tr e^{-\beta_H(E)H} $  and $h_1$ is the minimal eigenvalue of $H$.

By Proposition 1 in \cite{EC} the Gibbs condition (\ref{H-cond}) is equivalent to the following asymptotic property
\begin{equation}\label{H-cond-a}
  F_{H}(E)=o\shs(E)\quad\textrm{as}\quad E\rightarrow+\infty.
\end{equation}

For example, if $\,N\doteq a^\dag a\,$ is the number operator of a quantum oscillator then $F_{\hat{N}}(E)=g(E)$, where
\begin{equation}\label{g-def}
 g(x)=(x+1)\ln(x+1)-x\ln x,\;\, x>0,\qquad g(0)=0.
\end{equation}

We will often assume that
\begin{equation}\label{star}
  h_1=\inf\limits_{\|\varphi\|=1}\langle\varphi\vert H\vert\varphi\rangle=0.
\end{equation}

For a nonnegative function $f$ on $[0,1]$ we denote  its non-decreasing
envelope by $f^{\uparrow}$, i.e.
\begin{equation*}
  f^{\uparrow}(x)\doteq \sup_{t\in[0,x]} f(t),\quad x\in[0,1].
\end{equation*}
In particular, we will  use  the \emph{nondecreasing concave} continuous function
\begin{equation}\label{h+}
h^{\uparrow}(x)=\left\{\begin{array}{ll}
        h(x) &\textrm{if}\;\;  x\in[0,\frac{1}{2}]\\
        \ln2\quad& \textrm{if}\;\;  x\in(\frac{1}{2},1].
        \end{array}\right.
\end{equation}
where $\,h(x)=\eta(x)+\eta(1-x)\,$ is the binary entropy.

\section{The $\varepsilon$-smoothed Entanglement of Formation}

\subsection{Definition and basic  properties}

Applying the Mora-Piani-Briegel construction (described in \cite{MPB}) to the EoF  and taking definition (\ref{EF-d+}) into account  we obtain the $\varepsilon$-smoothed Entanglement of Formation ($\varepsilon$-smoothed EoF) defined
for an arbitrary state $\rho$ of a bipartite system $AB$ and  $\varepsilon\in[0,1]$
as follows
\begin{equation}
E_{F,\varepsilon}(\rho)
\doteq
\inf_{\substack{\sigma\in\St(\HH_{AB})\\T(\rho,\sigma)\leq\varepsilon}}
E_F(\sigma)=\inf_{\mu\in\P^{\shs\rm d}_{\rho,\varepsilon}(\mathcal{H}_{AB})}S_A(\mu)=\inf_{\mu\in\P_{\rho,\varepsilon}(\mathcal{H}_{AB})}S_A(\mu),
\label{SEF-def}
\end{equation}
where $S_A$ is the function defined in (\ref{SA-def}), $\P_{\rho,\varepsilon}(\mathcal{H}_{AB})$ is the subset of $\P(\HH_{AB})$ consisting of measures $\mu$ such that
$$
\bar{\rho}(\mu)\in U_\varepsilon(\rho)\doteq\{\sigma\in\St(\mathcal{H}_{AB})\,|\, T(\rho,\sigma)\leq \varepsilon\}.
$$
and $\P^{\shs\rm d}_{\rho,\varepsilon}(\mathcal{H}_{AB})$ is the subset of $\P_{\rho,\varepsilon}(\mathcal{H}_{AB})$ consisting of discrete measures $\mu$.
Note that, in contrast to  (\ref{EF-d+}), we cannot assert that the third infimum in (\ref{SEF-def}) is attained, since the set $\P_{\rho,\varepsilon}(\mathcal{H}_{AB})$
is not compact (in contrast to the set $\P_{\rho}(\mathcal{H}_{AB})$).\smallskip

In the following theorem we describe basic properties of the function $E_{F,\varepsilon}$.\smallskip

\begin{theorem}\label{P-1} \emph{Let $A$ and $B$ be quantum systems of any dimensions.}\smallskip

\noindent A) \emph{For a given $\varepsilon\in(0,1]$ and a state $\rho$ in $\St(\HH_{AB})$
the following inequality holds
\[
E_{F,\varepsilon}(\rho)\leq \ln n<+\infty,
\]
where \(n\) is the minimal number such that either
\[
\sum_{i>n}\lambda_i^A\leq\varepsilon^2
\qquad\text{or}\qquad
\sum_{i>n}\lambda_i^B\leq\varepsilon^2,
\]
where $\{\lambda_i^A\}_{i=1}^{+\infty}$ and
$\{\lambda_i^B\}_{i=1}^{+\infty}$ are the sequences of eigenvalues of
the states $\rho_A$ and $\rho_B$, respectively, arranged in
nonincreasing order.}\footnote{$\lambda_i^A$ and $\lambda_i^B$ may be equal to zero for large $i\,$ if either $\rho_A$ or $\rho_B$  is a finite-rank state.}    \smallskip

\noindent B) \emph{The equality  $\,E_{F,\varepsilon}(\rho)=0\,$ holds if and only if the closed $\varepsilon$-vicinity $U_\varepsilon(\rho)$ of the state $\rho$ contains separable states. Moreover,
\begin{equation}
 E_{F,\varepsilon}(\rho)
 \geq 2\big[d\!\left(\rho,\St_{\rm sep}^{AB}\right)-\varepsilon\big]_+^2\qquad ([x]_+\doteq\max\{x,0\}),
 \label{eq:EF-lower-bound}
\end{equation}
where
\begin{equation}\label{d+}
 d\!\left(\rho,\St^{AB}_{\rm sep}\right)=\inf_{\sigma\in\St_{\rm sep}^{AB}}T(\rho,\sigma)
\end{equation}
is the distance from a state $\rho\in\St(\mathcal H_{AB})$ to the set $\St_{\rm sep}^{AB}$ of all separable states.}\footnote{The proof of this claim is based on the fact that the infimum in (\ref{d+}) is attained in general infinite-dimensional settings (Lemma \ref{atd} in the Appendix).}\smallskip

\noindent C) \emph{If $\,E_F(\rho)<+\infty\,$ then}
\begin{equation*}\label{scb+}
E_{F,\varepsilon}(\rho)\leq (1-\varepsilon)E_F(\rho)\qquad \forall \varepsilon\in[0,1].
\end{equation*}

\noindent D) \emph{For an arbitrary  state $\rho\in\St(\HH_{AB})$ the function  $\,\varepsilon \mapsto E_{F,\varepsilon}(\rho)\in[0,+\infty]\,$ is continuous, non-increasing and convex on $[0,1]$. In particular,}
$$
E_{F,\varepsilon}(\rho)\nearrow E_{F,0}(\rho)\doteq E_F(\rho)\leq +\infty \quad\textit{ as }\quad \varepsilon\searrow 0^+\quad\emph{ and }\quad E_{F,1}(\rho)=0.
$$

\noindent E) \emph{If $\,E_F(\rho)<+\infty\,$ and $\,\varepsilon\in(0,1]\,$ then}
\begin{equation}\label{scb}
E_{F,\varepsilon}(\rho)-E_{F,\varepsilon}(\sigma)\leq \frac{T(\rho,\sigma)E_F(\rho)}{\varepsilon+T(\rho,\sigma)}\qquad \forall \sigma\in\St(\HH_{AB}).
\end{equation}

\noindent F) \emph{For a given $\,\varepsilon\in(0,1]$  the function
$E_{F,\varepsilon}$ is convex and  continuous on $\St(\HH_{AB})$. It is Lipschitz continuous on any
subset  $\,\St\subset\St(\HH_{AB})$ such that $\,M_\St\doteq\sup_{\rho\in \St}E_F(\rho)$ is finite:}
$$
|E_{F,\varepsilon}(\rho)-E_{F,\varepsilon}(\sigma)|\leq\frac{M_\St T(\rho,\sigma)}{\varepsilon+T(\rho,\sigma)}\leq\frac{M_\St}{\varepsilon}\,T(\rho,\sigma)\qquad \forall \rho,\sigma\in\St.
$$

\noindent G) \emph{Let $\,\rho$ be a given state in $\St(\HH_{AB})$. If  $\,\Phi$ is a  transformation of $\,\St(\HH_{AB})$ such that
$$
E_F(\Phi(\sigma))\leq E_F(\sigma)\quad\textit{ and }\quad \|\Phi(\rho)-\Phi(\sigma)\|_1\leq \|\rho-\sigma\|_1
$$
for all $\sigma$ in $\St(\HH_{AB})$ then
$\,E_{F,\varepsilon}(\Phi(\rho))\leq E_{F,\varepsilon}(\rho)\,$
for all  $\varepsilon\in[0,1]$.}\medskip

\noindent  H) \emph{Let $\,\rho$ be a given state in $\St(\HH_{AB})$. If $\,\Phi$ is a  transformation of $\,\St(\HH_{AB})$ such that
$$
\Phi(\rho)=\rho,\quad E_F(\Phi(\sigma))\leq E_F(\sigma)\quad\textit{ and }\quad \|\Phi(\rho)-\Phi(\sigma)\|_1\leq \|\rho-\sigma\|_1
$$
for all $\sigma$ in $\St(\HH_{AB})$ then for any $\varepsilon\in[0,1]$
\begin{itemize}
  \item the first infimum in (\ref {SEF-def}) can be taken over the set
$$
\Phi(\St(\HH_{AB}))\cap U_\varepsilon(\rho)\doteq \left\{\Phi(\varsigma)\,|\, \varsigma\in\St(\HH_{AB}), T(\rho,\Phi(\varsigma))\leq \varepsilon\right\}
$$
  \item the third and the second infima in (\ref {SEF-def}) can be taken, respectively, over the set
$$
 \left\{\mu\in\P_{\rho,\varepsilon}(\HH_{AB})\,|\, \bar{\rho}(\mu)\in\Phi(\St(\HH_{AB}))\cap U_\varepsilon(\rho)\right\}
$$
and its subset consisting of discrete measures.
\end{itemize}}
\end{theorem}\medskip

\emph{Proof.} A)
Let \(P_n^X\) be the spectral projector of the state \(\rho_X\)
corresponding to its \(n\) largest eigenvalues, taking multiplicities
into account, \(X=A,B\). Consider the states
\[
\sigma_n^A=
\frac{
(P_n^A\otimes I_B)\shs\rho\,(P_n^A\otimes I_B)
}{
\operatorname{Tr}P_n^A\rho_A
},
\qquad
\sigma_n^B=
\frac{
(I_A\otimes P_n^B)\shs\rho\,(I_A\otimes P_n^B)
}{
\operatorname{Tr}P_n^B\rho_B
}.
\]
Then using Lemma 9.4.1 in \cite{Wilde} we obtain
\[
\frac12\|\sigma_n^X-\rho\|_1
\leq
\sqrt{\operatorname{Tr}(I_X-P_n^X)\rho_X}
=
\sqrt{\sum_{i>n}\lambda_i^X},\quad X=A,B.
\]
Therefore, if $\,\sum_{i>n}\lambda_i^X\leq\varepsilon^2$, where $X$ is either $A$ or $B$, then
$\,E_{F,\varepsilon}(\rho)
\leq E_F(\sigma_n^X)
\leq \ln n.$
 \smallskip

B) It is clear that $\,E_{F,\varepsilon}(\rho)=0\,$ if $\,U_\varepsilon(\rho)\cap\St_{\rm sep}^{AB}\neq\varnothing$. So, to prove this claim it suffices to
show that
\begin{equation}\label{rin}
d\!\left(\rho,\St_{\rm sep}^{AB}\right)>\varepsilon
\end{equation}
and to verify the validity of inequality \eqref{eq:EF-lower-bound} in the case $\,U_\varepsilon(\rho)\cap\St_{\rm sep}^{AB}=\varnothing$.\footnote{The assumption $\,U_\varepsilon(\rho)\cap \mathfrak{A}=\varnothing\,$ for a closed convex non-compact subset  $\mathfrak{A}\subset\St_{\rm sep}(\HH_{AB})\,$ with $\,\dim\HH_{AB}=+\infty\,$ does not imply $\,d\!\left(\rho,\mathfrak{A}\right)>\varepsilon$.}

By Lemma \ref{atd} in the Appendix, the distance from a state $\rho$ to the set $\St_{\rm sep}^{AB}$ is attained. Consequently, the assumption
$U_\varepsilon(\rho)\cap\St_{\rm sep}^{AB}=\varnothing$ implies (\ref{rin}).

For any $\sigma\in U_\varepsilon(\rho)$, the triangle inequality gives
\begin{align*}
 d(\sigma,\St_{\rm sep}^{AB})\,
 &=\inf_{\omega\in\St_{\rm sep}^{AB}}T(\sigma,\omega)\geq\inf_{\omega\in\St_{\rm sep}^{AB}}
       \big(T(\rho,\omega)-T(\rho,\sigma)\big)\geq d(\rho,\St_{\rm sep}^{AB})-\varepsilon.
\end{align*}
It follows from the quantum Pinsker inequality that
\begin{align*}
 E_F(\sigma)
 &\geq E_R(\sigma)
 =\inf_{\omega\in\St_{\rm sep}^{AB}}D(\sigma\Vert\omega)\geq2
       \inf_{\omega\in\St_{\rm sep}^{AB}}T(\sigma,\omega)^2\\
 &=2d\!\left(\sigma,\St_{\rm sep}^{AB}\right)^2\geq2\big(d\!\left(\rho,\St_{\rm sep}^{AB}\right)-\varepsilon\big)^2>0
\end{align*}
where the first inequality follows from Remark \ref{imp} and  $D(\cdot\|\cdot)$ is the quantum relative entropy \cite{H-SCI,Wilde}. Taking the infimum over all $\sigma\in U_\varepsilon(\rho)$ proves \eqref{eq:EF-lower-bound}. Since $d(\rho,\St_{\rm sep}^{AB})>\varepsilon$, the lower bound is strictly positive.\smallskip

C) To prove this claim it suffices to take an arbitrary separable state $\tau$ and consider the state $\sigma=(1-\varepsilon)\rho+\varepsilon \tau$, since
$$
T(\rho,\sigma)\leq \varepsilon\quad\textrm{ and }\quad E_F(\sigma)\leq (1-\varepsilon)E_F(\rho)+\varepsilon E_F(\tau)=(1-\varepsilon)E_F(\rho),
$$
by the convexity of $E_F$ and because $E_F(\tau)=0$. \smallskip

D) Fix a state $\rho$ and write $\,r_\rho(\varepsilon)\doteq E_{F,\varepsilon}(\rho)\,$.
The function $r_\rho$ is nonincreasing and $r_\rho(1)=0$ as the set $U_1(\rho)$ contains (all) separable states. It is also convex. Indeed, following the proof of Theorem 4 in \cite{MPB} assume that $\sigma_1$ and $\sigma_2$ are states such that
$\,T(\rho,\sigma_i)\leq\varepsilon_i,$ $i=1,2$,
and put
\[
\sigma_t=t\sigma_1+(1-t)\sigma_2,
\qquad 0\leq t\leq1.
\]
Joint convexity of the trace distance and convexity of $E_F$ imply
\[
T(\rho,\sigma_t)
\leq
t\varepsilon_1+(1-t)\varepsilon_2
\quad
\textrm{and}
\quad
E_F(\sigma_t)
\leq
tE_F(\sigma_1)+(1-t)E_F(\sigma_2).
\]
Taking infima over $\varepsilon_1$ and $\varepsilon_2$ gives
\[
r_\rho(t\varepsilon_1+(1-t)\varepsilon_2)
\leq
t r_\rho(\varepsilon_1)+(1-t)r_\rho(\varepsilon_2).
\]

Convexity and monotonicity of $r_\rho(\varepsilon)$ were established before. Since this function
is finite on $(0,1]$ by claim $A$, a standard property of finite convex functions on subsets of $\mathbb{R}^1$
implies continuity of $r_\rho(\varepsilon)$ on $(0,1)$.  At the right endpoint, monotonicity of $r_\rho(\varepsilon)$ gives
\[
0=r_\rho(1)
\leq
\liminf_{\varepsilon\to1^-}r_\rho(\varepsilon)\leq\limsup_{\varepsilon\to1^-}r_\rho(\varepsilon)\leq r_\rho(1)=0,
\]
where the last inequality follows from the convexity of $r_\rho(\varepsilon)$. It remains to consider $\varepsilon=0$.  It is clear that
$\,\sup_{\varepsilon\in(0,1]}r_\rho(\varepsilon)\leq r_\rho(0)\doteq E_F(\rho).$ Hence,
$$
r_\rho(\varepsilon)\nearrow C\leq E_F(\rho)\leq +\infty \quad\textup{as}\quad \varepsilon\searrow 0^+.
$$
Assume that $\,C<E_F(\rho)$. This implies that $C<+\infty$.  Let
$\varepsilon_n\to0^+$ as $n\to+\infty$.  Choose
states $\sigma_n$ such that
\[
T(\rho,\sigma_n)\leq\varepsilon_n
\quad\textrm{ and }\quad
E_F(\sigma_n)
\leq
r_\rho(\varepsilon_n)+\frac1n\,.
\]
Then $\,\sigma_n\rightarrow\rho\,$ as $\,n\to+\infty\,$ in trace norm. The lower semicontinuity of $E_F$ therefore gives
\[
E_F(\rho)
\leq
\liminf_{n\to+\infty}E_F(\sigma_n)
\leq
\liminf_{n\to+\infty}
\left(r_\rho(\varepsilon_n)+\frac1n\right).
\]
This contradicts our assumption.\smallskip

E) This claim is proved by a simple modification of the arguments used in \cite{MPB} to prove continuity of the function $E_{\varepsilon}^{(D)}$ in the finite-dimensional settings.

Denote $\,T(\rho,\sigma)\,$ by $\,\delta$. Let $\theta>0$ be arbitrary and $\sigma_*$ be a  state
in $U_\varepsilon(\sigma)$ such that $E_{F,\varepsilon}(\sigma)\geq E_F(\sigma_*)-\theta$. Take
$\,\lambda=\frac{\varepsilon}{\varepsilon+\delta}$
and consider the state
\[
\omega=\lambda\sigma_*+(1-\lambda)\rho.
\]
Then
$$
\|\omega-\rho\|_1=\lambda\|\sigma_*-\rho\|_1\leq
\lambda\|\sigma_*-\sigma\|_1+\lambda\|\sigma-\rho\|_1\leq 2\lambda(\varepsilon+\delta)=2\varepsilon
$$
by the triangle inequality. Therefore,
\begin{align*}
E_{F,\varepsilon}(\rho)
&\leq E_F(\omega)\\
&\leq \lambda E_F(\sigma_*)
 +(1-\lambda)E_F(\rho)\\
&\leq \lambda E_{F,\varepsilon}(\sigma)
 +\lambda\theta+(1-\lambda)E_F(\rho),
\end{align*}
where we used the convexity of \(E_F\) and the choice of \(\sigma_*\).

Thus, using the expression for \(\lambda\), we obtain
\[
E_{F,\varepsilon}(\rho)-E_{F,\varepsilon}(\sigma)
\leq\frac{\delta}{\varepsilon+\delta}\left(E_F(\rho)-E_{F,\varepsilon}(\sigma)\right)+\frac{\varepsilon\theta}{\varepsilon+\delta}.
\]
Since $\theta$ is arbitrary and \(E_{F,\varepsilon}(\sigma)\geq0\), this inequality implies (\ref{scb}).\smallskip

F) The lower semicontinuity of the finite function  $E_{F,\varepsilon}$ on $\St(\HH_{AB})$ follows from Proposition \ref{mm}B in Section 3.2 below and the
continuity of the function $E^{\lambda}_{F}$ \cite[Theorem 1]{EF}. Thus, to prove the continuity of this function on $\St(\HH_{AB})$ it suffices to come to a contradiction  assuming existence of a sequence $\{\rho_n\}\subset\St(\HH_{AB})$ converging to a state $\rho_0\in\St(\HH_{AB})$ such that
\begin{equation}\label{ouras}
\exists\lim_{n\to+\infty}E_{F,\varepsilon}(\rho_n)=a>E_{F,\varepsilon}(\rho_0).
\end{equation}

By claims A and C the function $\,\varepsilon\mapsto E_{F,\varepsilon}(\rho_0)\,$ is finite and continuous on $(0,1]$. So, there is $\delta\in(0,\varepsilon)$ such that  $a>E_{F,\varepsilon-\delta}(\rho_0)$. But it is clear that there exists
$n_\delta$ such that $\,U_{\varepsilon}(\rho_n)\supseteq U_{\varepsilon-\delta}(\rho_0)\,$ and, hence, $E_{F,\varepsilon}(\rho_n)\leq E_{F,\varepsilon-\delta}(\rho_0)$
for all $n\geq n_\delta$. This contradicts (\ref{ouras}).\smallskip

The convexity of the function $E_{F,\varepsilon}$ is proved by the obvious modifications of the arguments from the proof of Theorem 2 in \cite{MPB}.\smallskip

G) Since $E_F^{\lambda}(\Phi(\rho))\leq E_F^{\lambda}(\rho)$ for any state $\rho$ and any $\lambda>0$  by Theorem 1F in \cite{EF}, by using
Proposition \ref{mm}B in Section 3.2 below we obtain
\[
E_{F,\varepsilon}(\Phi(\rho))
=
\sup_{\lambda>0}
\left\{
E_F^\lambda(\Phi(\rho))-\frac{\varepsilon}{\lambda}
\right\}
\leq
\sup_{\lambda>0}
\left\{
E_F^\lambda(\rho)-\frac{\varepsilon}{\lambda}
\right\}.
\]

H) This claim follows directly from Theorem 1G in \cite{EF} and Proposition \ref{mm}B in Section 3.2 below. $\Box$\medskip

\textbf{Note:} In the proofs of claims F, G and H of  Theorem \ref{P-1} we have used Proposition \ref{mm}B in Section 3.2 below which is established by using claim D
of  Theorem \ref{P-1} (proved independently of claims  F, G and H). \smallskip

Theorem \ref{P-1}G and the LOCC-monotonicity of the EoF imply the following observation (which, in finite dimensions,  follows from Theorem 1 in \cite{MPB})\smallskip

\begin{corollary}\label{E-r} \emph{For every $\varepsilon\in[0,1]$ the function $E_{F,\varepsilon}$ does not increase under nonselective
LOCC-operations (cf.\cite{Vidal,P&V,Lami-new}).}
\end{corollary}
\medskip

Note that $E_{F,\varepsilon}$ may increase under selective LOCC-operations. This can be shown using the general example in \cite[Proposition 2]{MPB}.  \smallskip

Theorem \ref{P-1}H allows us to independently prove a $E_F$-version of Proposition 1 in \cite{MPB}.\smallskip
\begin{corollary}\label{Edc}\emph{ For any $\varepsilon\in[0,1]$  the first (resp. the second, the third) infimum in (\ref{SEF-def}) can be taken over the set
$\,\St(\HH^\rho_A\otimes\HH^\rho_B)\,$ (resp. $\,\P^{\shs\rm d}_{\rho,\varepsilon}(\HH^\rho_A\otimes\HH^\rho_B),\, \P_{\rho,\varepsilon}(\HH^\rho_A\otimes\HH^\rho_B)$), where $\,\HH^\rho_X\doteq\supp\rho_X$, $X=A,B$.}
\end{corollary}\smallskip

To derive Corollary \ref{Edc} from Theorem \ref{P-1}G one should take the channel
\begin{equation}\label{ch-form}
\Phi(\sigma)=P_A\otimes P_B\cdot \sigma\cdot P_A\otimes P_B + [\Tr(I_{AB}-P_A\otimes P_B)\sigma]\shs \tau,\quad \sigma\in\St(\HH_{AB}),
\end{equation}
where $P_X$ is the projector onto $\HH^\rho_X$, $X=A,B$, and $\tau$ is an arbitrary separable state in $\,\St(\HH^\rho_A\otimes\HH^\rho_B)$.
\medskip

\begin{remark}\label{Edr}
Corollary \ref{Edc} shows the invariance of the definition of the function $E_{F,\varepsilon}$  w.r.t. to the embeddings $\,\HH_A\subset\HH_{A'}$ and $\,\HH_B\subset\HH_{B'}$:
for any $\rho\in\St(\HH_{AB})$ we have
$$
\inf_{\sigma\in\St(\HH_{AB})}\!\left\{E_F(\sigma)\,|\, T(\rho,\sigma)\leq \varepsilon\right\}=\inf_{\sigma\in\St(\HH_{A'B'})}\!\left\{E_F(\sigma)\,|\, T(\rho,\sigma)\leq \varepsilon\right\}\quad \forall\varepsilon\in(0,1).
$$
\end{remark}\medskip

\subsection{Relation to the Moreau--Yosida approximation of the EoF}

The proof of several claims of  Theorem \ref{P-1} and the important results below  are based of the relations between  the family  $\{E_{F,\varepsilon}\}_{\varepsilon\in[0,1]}$ of  $\varepsilon$-smoothed EoF
and the family $\{E_F^\lambda\}_{\lambda>0}$ of Moreau--Yosida approximations of the EoF  defined as
\begin{equation}\label{EFA+}
 E_F^\lambda(\rho)=\inf_{\sigma\in\St(\HH_{AB})}\left\{E_F(\sigma)+\frac{1}{\lambda}\,T(\rho,\sigma)\right\},
\end{equation}
which are  introduced and explored in \cite{EF}. These relations are described in the following\medskip

\begin{proposition}\label{mm} \emph{Let $A$ and $B$ be quantum systems of any dimensions.}\medskip

A) \emph{For every state $\rho$ in $\,\St(\HH_{AB})$ and every $\lambda>0$ the following relation holds
\begin{equation}
E_F^\lambda(\rho)
=
\inf_{0\leq\varepsilon\leq1}
\left\{
E_{F,\varepsilon}(\rho)+\frac{\varepsilon}{\lambda}
\right\}.
\label{eq:forward}
\end{equation}
The infimum in (\ref{eq:forward}) is always attained  (so, it can be replaced by minimum).} \medskip

B) \emph{For every state $\rho$ in $\,\St(\HH_{AB})$ and every $\varepsilon\in[0,1]\,$ the following relation holds
\begin{equation}
E_{F,\varepsilon}(\rho)
=
\sup_{\lambda>0}
\left\{
E_F^\lambda(\rho)-\frac{\varepsilon}{\lambda}
\right\}.
\label{eq:backward}
\end{equation}
The supremum  in (\ref{eq:backward}) is not attained, in general, even in the case $\HH_A=\HH_B=\mathbb{C}^2$.} \medskip
\end{proposition}

\emph{Proof.} Note first that the below proof exploits only those claims of Theorem \ref{P-1} in the proofs of which the Proposition \ref{mm} was not used. \smallskip

A) For a state $\sigma$, put
$\varepsilon=T(\rho,\sigma)$.  By the definition of the smoothed EoF,
\[
E_{F,\varepsilon}(\rho)\leq E_F(\sigma),
\]
and hence
\[
E_{F,\varepsilon}(\rho)+\frac{\varepsilon}{\lambda}
\leq
E_F(\sigma)+\frac{1}{\lambda}\,T(\rho,\sigma).
\]
Taking the infimum over all states $\sigma$  gives
\[
\inf_{\varepsilon\in[0,1]}
\left\{
E_{F,\varepsilon}(\rho)+\frac{\varepsilon}{\lambda}
\right\}\leq
E_F^\lambda(\rho).
\]
Conversely, if $T(\rho,\sigma)\leq\varepsilon$, then
\[
E_F^\lambda(\rho)
\leq
E_F(\sigma)+\frac{1}{\lambda}\,T(\rho,\sigma)
\leq
E_F(\sigma)+\frac{\varepsilon}{\lambda}.
\]
Taking the infimum over all states in the $\varepsilon$-ball $U_\varepsilon(\rho)$ and then
over $\varepsilon$ proves \eqref{eq:forward}.\smallskip

The last claim of A follows from continuity of the function  $\,\varepsilon \mapsto E_{F,\varepsilon}(\rho)+\frac{\varepsilon}{\lambda}\,$
for any given state $\rho$, which follows from Theorem \ref{P-1}D.
 \medskip

B) Fix a bipartite state $\rho$ and define $\,r_{\rho}(\varepsilon)\doteq E_{F,\varepsilon}(\rho)$, $\,0\leq \varepsilon\leq 1$.
Extend the function $r_{\rho}$ to the whole real line by setting
\[
r_{\rho}(\varepsilon)=+\infty
\qquad\text{for }\varepsilon\notin[0,1].
\]
By Theorem \ref{P-1}D this extended function  $r_{\rho}(\varepsilon)$  is proper, convex and lower
semicontinuous. Note also that $r_\rho$ is a finite nonincreasing function on $(0,1]$ and that $r_{\rho}(1)=0$.

The Fenchel conjugate of the function $r_{\rho}$ is
\begin{equation}
r_{\rho}^*(p)
=
\sup_{t\in\mathbb R}\{pt-r_{\rho}(t)\}
=
\sup_{0\leq t\leq 1}\{pt-r_{\rho}(t)\}.
\label{eq:fenchel-conjugate-r}
\end{equation}
By the Fenchel--Moreau theorem,
\begin{equation}
r_{\rho}(\varepsilon)
=
r_{\rho}^{**}(\varepsilon)
=
\sup_{p\in\mathbb R}\{p\varepsilon-r_{\rho}^*(p)\}.
\label{eq:fenchel-moreau-r}
\end{equation}

We want to show that the supremum in
\eqref{eq:fenchel-moreau-r} may be restricted to $p\leq 0$.
For $p>0$, the monotonicity of $r_{\rho}$ gives $r_{\rho}(t)\geq r_{\rho}(1)$ for every
$t\in[0,1]$. Hence
\[
pt-r_{\rho}(t)\leq p-r_{\rho}(1),
\]
and equality holds at $t=1$. Thus
\begin{equation}
r_{\rho}^*(p)=p-r_{\rho}(1),\qquad p>0.
\label{eq:conjugate-positive-p}
\end{equation}
Consequently, for $\varepsilon\in[0,1)$,
\[
p\varepsilon-r_{\rho}^*(p)
=
p(\varepsilon-1)+r_{\rho}(1)
< r_{\rho}(1)\leq r_{\rho}(\varepsilon),
\]
while for $\,\varepsilon=1\,$ we have $\,p\varepsilon-r_{\rho}^*(p)=p-r_{\rho}^*(p)\leq r_{\rho}(1)=0$.

Therefore, it follows form (\ref{eq:fenchel-moreau-r}) that
\begin{equation}
r_{\rho}(\varepsilon)
\geq
\sup_{p\leq0}\{p\varepsilon-r_{\rho}^*(p)\},\quad \forall\varepsilon\in[0,1],
\label{eq:negative-halfline-including-zero}
\end{equation}
and an equality holds in (\ref{eq:negative-halfline-including-zero}) for $\,\varepsilon\neq1$. In fact, it holds
for $\,\varepsilon=1\,$ as well, because
\[
-r_{\rho}^*(0)=\inf_{0\leq t\leq1}r_{\rho}(t)=r_{\rho}(1)=0.
\]

Since $\,p=p+r_{\rho}(1)\leq r_{\rho}^*(p)\leq 0\,$  for $\,p<0\,$ we have
\[
\lim_{p\to0^-}r_{\rho}^*(p)=r_{\rho}^*(0)=-r_{\rho}(1)=0.
\]
It follows that the value of $\,p\varepsilon-r_{\rho}^*(p)\,$ corresponding to $p=0$ can be obtained as a
limit along negative values of $p$. Thus, we conclude that
\begin{equation}
r_{\rho}(\varepsilon)
=
\sup_{p<0}\{p\varepsilon-r_{\rho}^*(p)\},
\qquad 0\leq\varepsilon\leq1.
\label{eq:negative-halfline}
\end{equation}

By  part A of the proposition  we have
\begin{equation}
E_F^\lambda(\rho)
=
\inf_{0\leq t\leq1}
\left\{
r_{\rho}(t)+\frac{t}{\lambda}
\right\}.
\label{eq:direct-transform-r}
\end{equation}
Using the definition of the Fenchel conjugate, we obtain
\begin{equation}
r_{\rho}^*\left(-\frac1\lambda\right)
=\sup_{0\leq t\leq1}
\left\{-\frac{t}{\lambda}-r_{\rho}(t)\right\}
=-\inf_{0\leq t\leq1}
\left\{r_{\rho}(t)+\frac{t}{\lambda}\right\}
=-E_F^\lambda(\rho).
\label{eq:conjugate-moreau}
\end{equation}

Now substitute $p=-1/\lambda$ into \eqref{eq:negative-halfline} using (\ref{eq:conjugate-moreau}).
As $p$ ranges over $\,(-\infty,0)$, the parameter
$\,\lambda=-1/p\,$ ranges over $\,(0,+\infty)$. Therefore,
\begin{equation}
r_{\rho}(\varepsilon)=\sup_{p<0}\{p\varepsilon-r_{\rho}^*(p)\}
=\sup_{\lambda>0}
\left\{
-\frac{\varepsilon}{\lambda}
-r_{\rho}^*\left(-\frac1\lambda\right)
\right\}
=\sup_{\lambda>0}
\left\{
E_F^\lambda(\rho)-\frac{\varepsilon}{\lambda}
\right\}.
\label{eq:inverse-transform}
\end{equation}
Recalling that $r_{\rho}(\varepsilon)=E_{F,\varepsilon}(\rho)$, we arrive at
the required  formula (\ref{eq:backward}).
\medskip

To prove the last claim of B consider the following example involving a
two-qubit maximally entangled state.

Let
\[
|\Omega_2\rangle
=
\frac{1}{\sqrt{2}}(|00\rangle+|11\rangle),
\qquad
\Omega_2=|\Omega_2\rangle\langle\Omega_2|,
\]
and consider the separable state
\[
\tau_2
=
\frac12|00\rangle\langle00|
+
\frac12|11\rangle\langle11|.
\]
Since $\,\Omega_2-\tau_2
=
\frac12\left(
|00\rangle\langle11|
+
|11\rangle\langle00|
\right)$, we have $\,T(\Omega_2,\tau_2)
=
\frac12\|\Omega_2-\tau_2\|_1
=
\frac12$. Consequently, $\,E_{F,\varepsilon}(\Omega_2)=0\,$
for every $\,\varepsilon\geq\frac12\,$
as $\,E_F(\tau_2)=0$. In particular,
\begin{equation}
E_{F,3/4}(\Omega_2)=0.
\label{eq:bell-smoothing-three-quarters}
\end{equation}

On the other hand, the definition of the Moreau--Yosida
approximation gives, for every $\lambda>0$,
\[
\begin{aligned}
E_F^\lambda(\Omega_2)
&=
\inf_{\sigma}
\left\{
E_F(\sigma)+\frac{T(\Omega_2,\sigma)}{\lambda}
\right\}\leq
E_F(\tau_2)+\frac{T(\Omega_2,\tau_2)}{\lambda}=\frac{1}{2\lambda}.
\end{aligned}
\]
It follows that
\[
E_F^\lambda(\Omega_2)-\frac{3}{4\lambda}
\leq
-\frac{1}{4\lambda}<0
\]
for every finite $\lambda>0$. At the same time  (\ref{eq:backward}) and
\eqref{eq:bell-smoothing-three-quarters} imply
\[
\sup_{\lambda>0}
\left\{
E_F^\lambda(\Omega_2)-\frac{3}{4\lambda}
\right\}
=E_{F,3/4}(\Omega_2)=0,
\]
while the expression inside the supremum is strictly negative
for every $\lambda>0$. Thus the supremum is not attained (it is achieved asymptotically as $\lambda\to +\infty$).

We have therefore established that the supremum in
(\ref{eq:backward}) cannot, in general, be
replaced by a maximum.   $\Box$

\subsection{Variational expression for the function $E_{F,\varepsilon}$}

The well known variational expression for the EoF (cf.\cite{Ulm,H-SCI,H&Sh-5}) and Proposition 1 in \cite{EF} state that
\begin{equation}\label{EV}
   E_F(\rho)=\sup_{\Lambda\in\mathfrak{D}(+\infty)}\Tr \Lambda\rho\quad \textrm{ and }\quad   E_F^{\lambda}(\rho)=\sup_{\Lambda\in\mathfrak{D}\!\left(\frac{1}{\lambda}\right)}\Tr \Lambda\rho,
\end{equation}
where
$$
\mathfrak{D}(d)\doteq\left\{\Lambda \in\B_{\rm sa}(\HH_{AB})\,\left|\, D(\Lambda)\leq d,\;\langle\psi|\Lambda|\psi\rangle
        \leq S(\psi_A)\;\, \forall\psi\in\HH^1_{AB}\right.\right\},\quad d\leq+\infty
$$
Here $\HH^1_{AB}$ is the unit sphere in $\HH_{AB}$ and  $D(\Lambda)$ is the diameter of the spectrum $\mathrm{Sp}(\Lambda)$ of an operator $\Lambda\in\B_{\rm sa}(\HH_{AB})$:
\begin{equation}\label{D-def}
D(\Lambda)=\max \mathrm{Sp}(\Lambda)-\min \mathrm{Sp}(\Lambda).
\end{equation}
It is easy to see that the condition $\,\langle\psi|\Lambda|\psi\rangle
        \leq S(\psi_A),\, \forall\psi\in\HH^1_{AB},$ used in the definition of the set $\mathfrak{D}(d)$ is equivalent to the following one
$$
\Tr\Lambda\sigma\leq E_F(\sigma)\quad \forall\sigma\in\St(\HH_{AB}).
$$
So, the set $\mathfrak{D}(d)$ can be equivalently defined as
$$
\mathfrak{D}(d)\doteq\left\{\Lambda \in\B_{\rm sa}(\HH_{AB})\,\left|\, D(\Lambda)\leq d,\; \Tr\Lambda\sigma\leq E_F(\sigma)\;\; \forall\sigma\in\St(\HH_{AB})\right.\right\}.
$$

Using the  second expression in (\ref{EV}) and the relations between  $E_{F,\varepsilon}$ and $E^{\lambda}_F$  stated in Proposition \ref{mm}  one can obtain the following\smallskip\newpage
\begin{proposition}\label{VEP}
\emph{Let $A$ and $B$ be quantum systems of any dimensions. For an arbitrary  state $\rho\in\St(\HH_{AB})$ and  $\varepsilon\in[0,1]$ the following expression holds
\begin{equation}\label{VE}
E_{F,\varepsilon}(\rho)= \sup_{\Lambda\in \mathfrak{F}(\varepsilon)}\Tr \Lambda\rho,
\end{equation}
where}
$$
\mathfrak{F}(\varepsilon)\doteq\left\{\Lambda \in\B_{\rm sa}(\HH_{AB})\,\left|\, \exists\lambda>0\textit{ s.t. }  D(\Lambda)\leq \frac{1}{\lambda},\; \langle\psi|\Lambda|\psi\rangle
        \leq S(\psi_A)-\frac{\varepsilon}{\lambda}\;\; \forall\psi\in\HH^1_{AB}\right.\right\}.
$$
\end{proposition}

\begin{remark}\label{D-r}
The set $\mathfrak{F}(\varepsilon)$ can be equivalently defined as
$$
\mathfrak{F}(\varepsilon)\doteq\left\{\Lambda \in\B_{\rm sa}(\HH_{AB})\,\left|\, \exists\lambda>0\textup{ s.t. }  D(\Lambda)\leq \frac{1}{\lambda},\; \Tr\Lambda\sigma\leq E_F(\sigma)-\frac{\varepsilon}{\lambda}\;\; \forall\sigma\in\St(\HH_{AB})\right.\right\}.
$$
It is clear that $\,\mathfrak{F}(0)=\mathfrak{D}(+\infty)\,$ and variational expression (\ref{VE}) turns into the variational expression for the EoF in (\ref{EV}).
\end{remark}\medskip

\emph{Proof.} By the above remark we may assume that $\varepsilon>0$. Proposition \ref{mm}B and the second expression in (\ref{EV}) imply
\begin{align*}
E_{F,\varepsilon}(\rho)
&=\,
\sup_{\lambda>0}
\left\{
E_F^\lambda(\rho)-\frac{\varepsilon}{\lambda}
\right\}=
\sup_{\lambda>0}
\sup_{\Lambda\in \mathfrak{D}(1/\lambda)}
\left\{
\operatorname{Tr}\Lambda\rho-\frac{\varepsilon}{\lambda}
\right\}\\
&=\,
\sup_{\lambda>0}
\sup_{\Lambda\in \mathfrak{D}(1/\lambda)}
\operatorname{Tr}
\left(\Lambda-\frac{\varepsilon}{\lambda}I_{AB}\right)\rho=
\sup_{\Lambda\in \mathfrak{F}(\varepsilon)}\Tr \Lambda\rho,
\end{align*}
because $\,D\!\left(\Lambda-\frac{\varepsilon}{\lambda}I_{AB}\right)=D(\Lambda)$. $\Box$

\subsection{Weak subadditivity of the $\varepsilon$-smoothed EoF}

Unfortunately, the function $E_{F,\varepsilon}$ does not inherit the subadditivity property of the EoF. This is shown by Example \ref{nsa} below.
Nevertheless, one can prove the following\smallskip

\begin{proposition}\label{wsa} \emph{Let $A_i$ and $B_i$, $i=1,2$, be quantum systems of any dimensions.  Let $\varepsilon_1$ and $\varepsilon_2$ be arbitrary nonnegative numbers such that $\varepsilon_1+\varepsilon_2\leq 1$. The inequality
\begin{equation}\label{wsa+}
E_{F,\varepsilon_1+\varepsilon_2}(\rho\otimes\sigma)
\leq
E_{F,\varepsilon_1}(\rho)
+
E_{F,\varepsilon_2}(\sigma)
\end{equation}
holds for any  states $\rho$ in $\,\St(\HH_{A_1B_1})\,$ and $\sigma$ in $\,\St(\HH_{A_2B_2}),$ where $\rho\otimes\sigma$ is treated as a state of the bipartite system $(A_1A_2)(B_1B_2)$.  }
\end{proposition}\medskip

\emph{Proof.} By Proposition \ref{mm}B  we have
\begin{align*}
E_{F,\varepsilon_1+\varepsilon_2}(\rho\otimes\sigma)
&=
\sup_{\lambda>0}
\left\{
E_F^\lambda(\rho\otimes\sigma)
-\frac{\varepsilon_1+\varepsilon_2}{\lambda}
\right\}\\
&\leq
\sup_{\lambda>0}
\left\{
E_F^\lambda(\rho)-\frac{\varepsilon_1}{\lambda}
\right\}
+
\sup_{\lambda>0}
\left\{
E_F^\lambda(\sigma)-\frac{\varepsilon_2}{\lambda}
\right\}=
E_{F,\varepsilon_1}(\rho)
+
E_{F,\varepsilon_2}(\sigma),
\end{align*}
where the inequality is due to the subadditivity of the function $E_F^\lambda$ established in \cite[Proposition 2]{EF}. $\Box$\medskip

\begin{remark}\label{wsa-r} Inequality (\ref{wsa+}) can be directly deduced from the definition of the function $E_{F,\varepsilon}$ by using the inequality
\begin{equation*}
\|\rho\otimes\sigma-\rho'\otimes\sigma'\|_1
\leq
\|\rho-\rho'\|_1+\|\sigma-\sigma'\|_1
\end{equation*}
valid for arbitrary states $\rho,\rho'\in \St(\HH_{A_1B_1})$ and  arbitrary states $\sigma,\sigma'\in \St(\HH_{A_2B_2})$. The above proof relates the weak subadditivity of the function $E_{F,\varepsilon}$ with
the subadditivity of MY-approximation $E_F^\lambda$ via  expression (\ref{eq:backward}). It is easy to see that
using expression (\ref{eq:forward}) one can deduce the subadditivity of  $E_F^\lambda$ from the  weak subadditivity of $E_{F,\varepsilon}$.
\end{remark}\smallskip

The following example shows that the weak subadditivity inequality (\ref{wsa+}) does not imply subadditivity at a fixed smoothing radius $\varepsilon$.
In fact, it shows that the latter property fails even for two-qubit maximally entangled states.\smallskip

\begin{example}\label{nsa} For $d\geq 2$, let
\[
|\Omega_d\rangle
=
\frac{1}{\sqrt d}
\sum_{j=0}^{d-1}|j\rangle_A|j\rangle_B,
\qquad
\Omega_d=|\Omega_d\rangle\langle\Omega_d|.
\]
We will show that
\begin{equation}
E_{F,1/2}(\Omega_2)=0,
\qquad
E_{F,1/2}(\Omega_2\otimes\Omega_2)>0,
\label{eq:fixed-radius-counterexample}
\end{equation}
where $\,\Omega_2\otimes\Omega_2\,$ is treated as a state of the bipartite system  $(A_1A_2)(B_1B_2)$.

First, consider the separable state
\[
\tau_2
=
\frac12|00\rangle\langle00|
+
\frac12|11\rangle\langle11|.
\]
Since $E_F(\tau_2)=0$, it suffices to calculate its trace
distance from $\Omega_2$. We have
\[
\Omega_2-\tau_2
=
\frac12\left(
|00\rangle\langle11|
+
|11\rangle\langle00|
\right).
\]
The nonzero eigenvalues of this operator are $1/2$ and
$-1/2$. Therefore,
\[
T(\Omega_2,\tau_2)
=
\frac12\|\Omega_2-\tau_2\|_1
=
\frac12.
\]
By the definition of the smoothed EoF we have
\begin{equation}
E_{F,1/2}(\Omega_2)=0.
\label{eq:bell-smoothed-zero}
\end{equation}

We next use the well known fact that the distance from a maximally entangled
state to the set $\St_{\rm sep}(\mathbb C^d\otimes\mathbb C^d)$ of separable states on
$\mathbb C^d\otimes\mathbb C^d$ is determined by the formula\footnote{I would be grateful for a direct reference on this formula.}
\begin{equation}
\inf_{\omega\in\St_{\rm sep}(\mathbb C^d\otimes\mathbb C^d)}
T(\Omega_d,\omega)
=
1-\frac1d.
\label{eq:distance-max-ent-separable}
\end{equation}

To prove (\ref{eq:distance-max-ent-separable}) note that for every pair of unit vectors $\,|\alpha\rangle=(a_0,...,a_{d-1})^\bot$ and
$\,|b\rangle=(b_0,...,b_{d-1})^\bot$, the Cauchy--Schwarz inequality gives
\[
\begin{aligned}
|\langle\Omega_d|\alpha\otimes \beta\rangle|^2
&=
\frac1d
\left|\sum_{j=0}^{d-1}a_jb_j\right|^2\leq \frac1d.
\end{aligned}
\]
By convexity, every $\omega\in\St_{\rm sep}(\mathbb C^d\otimes\mathbb C^d)$ satisfies
\[
\langle\Omega_d|\omega|\Omega_d\rangle
\leq\frac1d.
\]
Since the spectral diameter of the operator $\Omega_d$ is equal to $1$, we have
\[
T(\Omega_d,\omega)
\geq
\left|\operatorname{Tr}\Omega_d(\Omega_d-\omega)\right|=
1-\langle\Omega_d|\omega|\Omega_d\rangle\geq 1-\frac1d.
\]

To prove that this lower bound is attained, consider the state
\[
\tau_d
=
\frac1d\sum_{j=0}^{d-1}|jj\rangle\langle jj|.
\]
On the subspace spanned by the vectors $|jj\rangle$,
the operator $\Omega_d-\tau_d$ has eigenvalue $(d-1)/d$
on $|\Omega_d\rangle$ and eigenvalue $-1/d$ with multiplicity
$d-1$. It vanishes on the orthogonal complement of this
subspace. Consequently, $\,T(\Omega_d,\tau_d)=1-\frac1d$,
which proves \eqref{eq:distance-max-ent-separable}.

Under the canonical reordering of tensor factors,
$\Omega_2\otimes\Omega_2$ is the maximally entangled state
$\Omega_4$ with respect to the bipartition $(A_1A_2)(B_1B_2)$.
So, it follows from \eqref{eq:distance-max-ent-separable} that
\begin{equation}
\inf_{\omega\in\St_{\rm sep}((\mathbb C^2\otimes\mathbb C^2)\otimes(\mathbb C^2\otimes\mathbb C^2))}
T(\Omega_2\otimes\Omega_2,\omega)=\inf_{\omega\in\St_{\rm sep}(\mathbb C^4\otimes\mathbb C^4)}
T(\Omega_4,\omega)=\frac34.
\label{eq:distance-phi4-separable}
\end{equation}

Hence it remains to prove that
\[
E_{F,1/2}(\Omega_2\otimes\Omega_2)=E_{F,1/2}(\Omega_4)>0.
\]

For every
state $\,\sigma \in U_{1/2}(\Omega_4)\doteq
\left\{
\sigma\in\St(\mathbb C^4\otimes\mathbb C^4)\,|\,
T(\Omega_4,\sigma)\leq\frac12
\right\}\,$
 and every state
$\omega\in\St_{\rm sep}(\mathbb C^4\otimes\mathbb C^4)$, the triangle inequality
and \eqref{eq:distance-phi4-separable} give
\[
T(\sigma,\omega)
\geq T(\Omega_4,\omega)-T(\Omega_4,\sigma)\geq \frac34-\frac12=\frac14.
\]
Thus, the set  $U_{1/2}(\Omega_4)$ contains no separable states.

Since the underlying Hilbert spaces are finite-dimensional,
the set $U_{1/2}(\Omega_4)$ is compact and the EoF is continuous.
Moreover, the EoF is nonnegative and vanishes precisely on
separable states. It follows that the minimum of $E_F$ on
$U_{1/2}(\Omega_4)$ is strictly positive. Therefore,
\begin{equation}
E_{F,1/2}(\Omega_4)
=
\min_{\sigma\in U_{1/2}(\Omega_4)}E_F(\sigma)>0.
\label{eq:positive-smoothed-phi4}
\end{equation}

Combining \eqref{eq:bell-smoothed-zero} and
\eqref{eq:positive-smoothed-phi4}, we obtain
\[
E_{F,1/2}(\Omega_2\otimes\Omega_2)
>
E_{F,1/2}(\Omega_2)
+
E_{F,1/2}(\Omega_2).
\]
Consequently, the smoothed EoF is not
subadditive when the same smoothing radius $\varepsilon$ is used for both
factors and their tensor product.

This counterexample is consistent with the weak subadditivity
inequality (\ref{wsa+}). Indeed, setting
$\,\varepsilon_1=\varepsilon_2=1/2\,$ in that inequality gives
\[
0=E_{F,1}(\Omega_2\otimes\Omega_2)
\leq
E_{F,1/2}(\Omega_2)
+
E_{F,1/2}(\Omega_2)
=0.
\]
\end{example}\smallskip

\begin{corollary}\label{tau} \emph{Let $A_i$ and $B_i$, $i=1,2$, be quantum systems of any dimensions and $\varepsilon\in[0,1]$. Then
\begin{equation}\label{tau+}
E_{F,\varepsilon}(\rho\otimes\sigma)=E_{F,\varepsilon}(\rho)
\end{equation}
for any  state $\rho$ in $\,\St(\HH_{A_1B_1})\,$ and any separable state $\sigma$ in $\,\St(\HH_{A_2B_2}),$ where $\rho\otimes\sigma$ is treated as a state of the bipartite system $(A_1A_2)(B_1B_2)$.}
\end{corollary}\smallskip

\emph{Proof.}  Proposition \ref{wsa} with $\,\varepsilon=\varepsilon_1\,$ and $\,\varepsilon_2=0\,$ implies that $"\leq"$ holds in (\ref{tau+}), since $E_F(\sigma)=0$. The converse inequality
follows from Corollary \ref{E-r} in Section 3.1. $\Box$

\subsection{Upper bounds on $E_F(\rho)-E_{F,\varepsilon}(\rho)$}

By Theorem \ref{P-1}D the family $\left\{E_{F,\varepsilon}\right\}_{\varepsilon\in(0,1]}$
of $\varepsilon$-smoothed EoF  pointwise converges to the function
$E_F$ as $\,\varepsilon\to0^+$. In this section we obtain
easily computable estimates for the rate of this convergence by using the (lower) semicontinuity bounds for the function $E_F$
obtained in~\cite{LCB+,AOE,EF}.\smallskip

The results below are based on the following obvious lemma.\smallskip

\begin{lemma}\label{vil}
\emph{Let $A$ and $B$ be quantum systems of any dimensions.
Let $\rho$ be a state of the system $AB$ with finite $E_F(\rho)$. If there exists a function $B_T$ on $[0,1]$ such that
$\,E_F(\rho)-E_F(\sigma)\leq B_T(\varepsilon)\,$
for all states $\,\sigma\in\St(\mathcal{H}_{AB})$ satisfying the condition $\,T(\rho,\sigma)\leq\varepsilon\in[0,1]\,$ then}
$$
   E_F(\rho)-E_{F,\varepsilon}(\rho)\leq B_T(\varepsilon)\qquad \forall \varepsilon\in[0,1].
$$
\end{lemma}
\medskip

By using  Lemma \ref{vil} and the semicontinuity bounds for the function $E_F$
from  Proposition 9 in \cite{AOE} and Proposition 10 in \cite{LCB+}  we obtain the following

\smallskip

\begin{proposition}\label{P-2} \emph{Let $A$ and $B$ be quantum systems of any dimensions and  $\rho$  a state of the system $AB$. Let  $\,\varepsilon\in(0,1]\,$ be arbitrary  and  $\,\delta_\varepsilon\doteq\sqrt{\varepsilon(2-\varepsilon)}$. Let  $\,h^\uparrow$ be non-decreasing
envelope of the binary entropy defined in (\ref{h+}).}
\smallskip

A) \emph{If $\,2\leq\rank \rho_A = d <+\infty,$ then}
\begin{equation}\label{pt-1}
    E_F(\rho)-E_{F,\varepsilon}(\rho)
    \leq \delta_\varepsilon\ln(d-1)+h^{\uparrow}(\delta_\varepsilon).
\end{equation}

B) \emph{If $\,\Tr H\rho_A = E < +\infty\,$ for some positive operator $H$ on $\HH_A$ satisfying conditions (\ref{H-cond}) and (\ref{star}), then
\begin{equation}\label{pt-2}
    E_F(\rho)-E_{F,\varepsilon}(\rho)
    \leq
    \delta_\varepsilon F_H\!\left(\frac{E}{\delta_\varepsilon}\right)+h^{\uparrow}(\delta_\varepsilon),
\end{equation}
where $F_H$ is the function defined in (\ref{F-def}). The r.h.s. of  (\ref{pt-2}) tends to zero as $\,\varepsilon\to0^+$.\footnote{The r.h.s. of (\ref{pt-2}) tends to zero as $\,\varepsilon\to0^+$ due to the equivalence of (\ref{H-cond}) and (\ref{H-cond-a}).}
}\end{proposition}
\medskip

\textbf{Note:} Claim B of Proposition \ref{P-2} is applicable to any state $\rho$ in $\St(\HH_{AB})$ such that $\,S(\rho_X)<+\infty,$  where $X$ is either $A$ or $B$, since for any such state there exists positive operators $H$ on $\HH_X$ satisfying  conditions (\ref{H-cond}) and (\ref{star}) such that $\,\Tr H\rho_X<+\infty$  \cite[Proposition 4]{EC}.\medskip

Proposition \ref{P-2} implies the following upper bounds on the rate of uniform convergence of
the $\varepsilon$-smoothed EoF  to the EoF as $\varepsilon\to0^+$ on the set of states with bounded rank/energy  of one of the marginal states.\smallskip

\begin{corollary}\label{C-1} \emph{The family $\{E_{F,\varepsilon}\}_{\varepsilon\in(0,1]}$
of $\varepsilon$-smoothed EoF  uniformly converges to the function $E_F$ on
the sets
\[
    \C^A_d
    =
    \{\rho\in\St(\mathcal{H}_{AB})\, |\,
    \rank\rho_A \leq d<+\infty\}
\]
and
\[
    \C^A_{H,E}
    =
    \{\rho\in\St(\mathcal{H}_{AB})\, |\,
    \Tr H\rho_A \leq E\},
\]
where $H$ is any positive operator on $\HH_A$ satisfying  conditions (\ref{H-cond}) and (\ref{star}). Moreover,
\begin{equation}\label{ub1}
    \sup_{\rho\in \C^A_d}
    \left|E_F(\rho)-E_{F,\varepsilon}(\rho)\right|
    \leq  \delta_\varepsilon\ln(d-1)
        +h^\uparrow(\delta_\varepsilon)\qquad (d\geq 2)
\end{equation}
and
\begin{equation}\label{ub2}
    \sup_{\rho\in\C^A_{H,E}}
    \left|E_F(\rho)-E_{F,\varepsilon}(\rho)\right|
    \leq  \delta_\varepsilon F_H\!\left(\frac{E}{\delta_\varepsilon}\right)
        + h^\uparrow(\delta_\varepsilon),
\end{equation}
where $F_H$ and $\,h^\uparrow$ are the functions defined  in (\ref{F-def}) and (\ref{h+}), respectively,  and $\,\delta_\varepsilon\doteq\sqrt{\varepsilon(2-\varepsilon)}$.}
\emph{The r.h.s. of  (\ref{ub2}) tends to zero as $\,\varepsilon\to0^+$.}
\end{corollary}
\medskip\medskip

The universal bounds given by Proposition \ref{P-2} can be improved for those states $\rho$ for which
there exist more accurate semicontinuity bounds than  the universal semicontinuity bounds
used in the proof of Proposition \ref{P-2}. In particular, for a state $\rho$ satisfying the condition of one of  cases A,B and C in Section 5 in \cite{EF}
Lemma \ref{vil} gives the bound
$$
E_F(\rho)-E_{F,\varepsilon}(\rho)
    \leq C_\rho\,\varepsilon,
$$
where $C_\rho$ is a characteristic of the spectrum of $\rho$.

\section{Regularization of the $\varepsilon$-smoothed Entanglement
of Formation and its relation to the $\varepsilon$-smoothed
Entanglement Cost}
\label{sec:regularization-smoothed-eof}

\subsection{Definitions and general remarks}

It was proved many years ago (cf.\cite{HaydenHorodeckiTerhal2001})  that the asymptotic entanglement cost
$E_C(\rho)$  of a state of a finite-dimensional bipartite quantum system is equal to the regularized EoF of this state defined as
\begin{equation}
E_F^\infty(\rho)
=
\lim_{n\to+\infty}\frac1n E_F(\rho^{\otimes n}),
\label{eq:regularized-eof}
\end{equation}
with the tensor product understood with respect to the bipartition
$(A^n)(B^n)$.
Consequently, in this setting we may, following to the general Mora-Piani-Briegel construction \cite{MPB} define
the $\varepsilon$-smoothed
Entanglement Cost as the minimal asymptotic entanglement cost among states
within trace distance $\varepsilon$ of a state $\rho$:
\begin{equation}\label{ECD}
E_{C,\varepsilon}(\rho)=
\inf_{\sigma:T(\rho,\sigma)\leq\varepsilon} E_C(\sigma)=\inf_{\sigma:T(\rho,\sigma)\leq\varepsilon} E_F^\infty(\sigma).
\end{equation}

The identification of $E_C$ with $E_F^\infty$ was established very recently for a state
of infinite-dimensional  bipartite quantum system if at least one of its  marginal entropies is finite
\cite{Lami-new}. This result, however, does not imply directly the second equality in (\ref{ECD}), because
$\varepsilon$-vicinity of any state of infinite-dimensional  bipartite quantum system contains states with both marginal entropies equal to $+\infty$ .

Keeping these results in mind, in this section, we compare two procedures:
first smoothing the EoF and then regularizing it, and first
regularizing the EoF and then smoothing the resulting functional.
We also compare these constructions with the operational
problem of preparing an approximate tensor power by LOCC.

We define the regularization of the smoothed EoF by
\begin{equation}
E_{F,\varepsilon}^{\infty}(\rho)
=
\limsup_{n\to+\infty}
\frac1n\, E_{F,\varepsilon}(\rho^{\otimes n}).
\label{eq:R-epsilon}
\end{equation}
The use of the upper limit is deliberate: in Section 3.4 it is shown that the subadditivity of
$E_F$ does not imply subadditivity of $E_{F,\varepsilon}$ at
a fixed $\varepsilon\in(0,1)$ and, hence, we cannot apply the standard arguments based on the Fekete lemma.\footnote{If $E_F(\rho)<+\infty$ then the  weak subadditivity and the monotonicity of
$E_{F,\varepsilon}$ under non-selective LOCC operations allow us to show that the sequence
of nonnegative numbers $\,a_n=\frac{1}{n}E_{F,\varepsilon}(\rho^{\otimes n})\,$ satisfies the condition
\[
a_n\leq \left(1+\frac{1}{n}\right)a_{n+1}
\leq a_n+\frac{c}{n},\quad\textrm{ where} \quad c=E_F(\rho).
\]
Unfortunately, this condition does not imply convergence of the sequence $\{a_n\}$ to a particular limit (one can construct a non-converging sequence $\{a_n\}$ of positive numbers satisfying this condition), but
it allows us to prove that all partial limits of this sequence  always form a certain  interval $\,[l,u]\,$  (it means that any $\,x\in[l,u]\,$ is a limit of some subsequence
of the sequence  $\{a_n\}$).}

By Theorem \ref{P-1}A for every $\varepsilon\in(0,1]$ and every state $\varrho$, we have
$\,0\leq E_{F,\varepsilon}(\varrho)<+\infty$.
It follows that each term in \eqref{eq:R-epsilon} is finite
when $\varepsilon>0$. Nevertheless, if $\,E_{F}^{\infty}(\rho)=+\infty,$ then the upper limit may be
infinite, and, in general, we can only conclude that
$\,E_{F,\varepsilon}(\rho)\in[0,+\infty]$.
At the endpoints,
\[
E_{F,0}^{\infty}(\rho)=E_F^\infty(\rho),
\qquad
E_{F,1}^{\infty}(\rho)=0.
\]
The first identity follows from $E_{F,0}=E_F$ and
\eqref{eq:regularized-eof}. For the second, every state lies
at trace distance at most one from every other state, so the
$\varepsilon$-vicinity of any state $\rho$ with $\varepsilon=1$ contains a separable state.

Theorem \ref{P-1}B implies that
\begin{equation}
E_{F,\varepsilon}^{\infty}(\rho)
\leq
(1-\varepsilon)E_F^\infty(\rho)\leq +\infty,
\qquad 0\leq\varepsilon<1,
\label{eq:R-epsilon-upper-bound}
\end{equation}

We define the smoothed regularized EoF by
\begin{equation}
[E^\infty_{F}]_\varepsilon(\rho)
=
\inf_{\substack{\sigma\in\mathfrak S(\mathcal H_{AB})\\
T(\rho,\sigma)\leq\varepsilon}}
E_F^\infty(\sigma).
\label{eq:R-tilde-epsilon}
\end{equation}
Thus $E_{F,\varepsilon}^{\infty}$ first smooths the EoF on the tensor
power and then takes the asymptotic rate, whereas
$[E^\infty_{F}]_\varepsilon$ first regularizes the EoF and
then minimizes the resulting functional over the trace-distance
ball.

The quantity $[E^\infty_{F}]_\varepsilon$ is well defined as an
extended nonnegative real number. In fact, for every
$\varepsilon>0$, the finite-dimensional truncation argument
used in the proof of Theorem \ref{P-1}A  provides a state $\sigma$ in the $\varepsilon$-ball
with finite EoF. Since
$\,E_F^\infty(\sigma)\leq E_F(\sigma)$,
we obtain
\[
0\leq [E^\infty_{F}]_\varepsilon(\rho)<+\infty
\qquad (\textrm{if } \varepsilon>0).
\]
At $\varepsilon=0$,
\[
[E^\infty_{F}]_0(\rho)=E_F^\infty(\rho)\leq+\infty.
\]
At $\varepsilon=1$, the $\varepsilon$-vicinity of any state $\rho$ contains separable
states, and hence $\,[E_{F}]_1^{\infty}(\rho)=0$.

\subsection{An example showing that $E^{\infty}_{F,\varepsilon}(\rho)>E_{C,\varepsilon}(\rho)=[E^\infty_{F}]_\varepsilon(\rho)$}

We now give an example in which $[E^\infty_{F}]_\varepsilon(\rho)$ and
$E_{F,\varepsilon}^{\infty}(\rho)$ differ strictly. Since this example
is in the finite-dimensional settings (in which  $[E^\infty_{F}]_\varepsilon(\rho)=E_{C,\varepsilon}(\rho)$), it also shows that
$E^{\infty}_{F,\varepsilon}(\rho)\neq E_{C,\varepsilon}(\rho)$ in general.

For an integer $d\geq2$, let
\[
|\Omega_d\rangle
=
\frac1{\sqrt d}\sum_{j=0}^{d-1}|j\rangle_A|j\rangle_B,
\qquad
\Omega_d=|\Omega_d\rangle\langle\Omega_d|,
\]
and define the separable state
\[
\tau_d
=
\frac1d\sum_{j=0}^{d-1}|jj\rangle\langle jj|.
\]
The operator $\Omega_d-\tau_d$ has eigenvalue $(d-1)/d$ on
$|\Omega_d\rangle$, eigenvalue $-1/d$ with multiplicity $d-1$
on the orthogonal complement of $|\Omega_d\rangle$ in the span
of $\{|jj\rangle\}_{j=0}^{d-1}$, and zero elsewhere. Hence
\begin{equation}
T(\Omega_d,\tau_d)=1-\frac1d.
\label{eq:Phi-tau-distance}
\end{equation}

Taking $\varepsilon=1-1/d$, the state $\tau_d$ is admissible
in \eqref{eq:R-tilde-epsilon}. Since it is separable, we have
$\,E_F^\infty(\tau_d)=0$. Therefore,
\begin{equation}
[E^\infty_{F}]_{1-1/d}(\Omega_d)=0.
\label{eq:R-tilde-Phi}
\end{equation}

We next prove the exact identity
\begin{equation}
E_{F,\varepsilon}^{\infty}(\Omega_d)=
(1-\varepsilon)\ln d,
\qquad 0\leq\varepsilon\leq1.
\label{eq:R-Phi-exact}
\end{equation}

Let $D=d^n$. Theorem \ref{P-1}C implies
\begin{equation}
E_{F,\varepsilon}(\Omega_D)
\leq(1-\varepsilon)E_{F}(\Omega_D)=(1-\varepsilon)\ln D.
\label{eq:Phi-upper}
\end{equation}

To prove the converse bound we will use the relative entropy of
entanglement
\[
E_R(\omega)
=
\inf_{\gamma\in\St_{\rm sep}(\mathbb C^{D}\otimes\mathbb C^{D})}
D(\omega\Vert\gamma),\quad \omega\in \St(\mathbb C^{D}\otimes\mathbb C^{D}),
\]
where $\,D(\omega\Vert\gamma)
=
\operatorname{Tr}\omega(\ln\omega-\ln\gamma)\,$
is the quantum relative entropy, and
$\St_{\rm sep}(\mathbb C^{D}\otimes\mathbb C^{D})$ is the set of separable states on
$\mathbb C^D\otimes\mathbb C^D$. We will use the standard inequality
\begin{equation}\label{FR}
E_F(\omega)\geq E_R(\omega)
\end{equation}
and the data-processing inequality for quantum relative
entropy \cite[Lemma 4]{L-2},\cite{Wilde}.

Let $\omega$ be any state satisfying
$T(\omega,\Omega_D)\leq\varepsilon$, and put $q=\operatorname{Tr}\Omega_D\omega$.
Since the spectral diameter of the operator $\Omega_D$ is equal to $1$, we have $\Tr\Omega_D(\Omega_D-\omega)\leq \varepsilon$ and hence
\[
q\geq1-\varepsilon.
\]
For any $\gamma\in\St_{\rm sep}(\mathbb C^{D}\otimes\mathbb C^{D})$, put
$p=\operatorname{Tr}\Omega_D\gamma$. The overlap of a maximally
entangled vector with any product vector is at most $1/D$ (by the arguments after (\ref{eq:distance-max-ent-separable})).
So, by convexity, $\,p\leq\frac1D$.

By applying the data-processing inequality to the two-outcome
measurement\break $\{\Omega_D,I-\Omega_D\}$ followed by $\Tr(\cdot)$ we obtain
\[
\begin{array}{c}
\displaystyle D(\omega\Vert\gamma)
\geq D(\{q,1-q\}\Vert\{p,1-p\})=
q\ln\frac{q}{p}
+(1-q)\ln\frac{1-q}{1-p}\\\\\geq q\ln D-h(q)\geq (1-\varepsilon)\ln D-\ln2,
\end{array}\]
where $h$ is the binary entropy. Taking the infimum over separable states $\gamma$ and using
(\ref{FR}) gives
\[
E_F(\omega)\geq
(1-\varepsilon)\ln D-\ln2.
\]
Since this holds for every state in the $\varepsilon$-vicinity of $\Omega_D$,
\begin{equation}
E_{F,\varepsilon}(\Omega_D)
\geq(1-\varepsilon)\ln D-\ln 2.
\label{eq:Phi-lower}
\end{equation}

Combining \eqref{eq:Phi-upper} and \eqref{eq:Phi-lower},
dividing by $n$, and using $\ln D=n\ln d$, we conclude that
\[
\lim_{n\to\infty}
\frac1n\, E_{F,\varepsilon}(\Omega_d^{\otimes n})
=
(1-\varepsilon)\ln d.
\]
This proves \eqref{eq:R-Phi-exact}.\smallskip

Thus,  equations
\eqref{eq:R-tilde-Phi} and \eqref{eq:R-Phi-exact} give
\begin{equation*}
E_{F,1-1/d}(\Omega_d)
=
\frac{\ln d}{d}
>
0
=
[E^\infty_{F}]_{1-1/d}(\Omega_d).
\end{equation*}
Therefore, smoothing the EoF before regularization and smoothing
the regularized EoF are not equivalent in general.\medskip

The above example  and the  unsuccessful attempts to find an example with an inverse inequality lead to the following open question.

\medskip
\noindent
\textbf{Open question.}
For every bipartite state $\rho$ on a tensor product of
separable Hilbert spaces and every $0<\varepsilon<1$, is it
true that
\begin{equation*}
E^\infty_{F,\varepsilon}(\rho)
\geq
[E^\infty_{F}]_\varepsilon(\rho)?
\end{equation*}

\subsection{Relation to the operational cost at fixed error}

We now distinguish the preceding state-space smoothing from
an operational preparation task with a fixed error tolerance.
For $0\leq\varepsilon<1$, define
\begin{equation}
C_n^\varepsilon(\rho)
=
\inf\left\{
\ln M\,\left|\,
\begin{array}{l}
\text{there exists a deterministic LOCC channel }\Lambda_n\\
\text{such that }
T\bigl(\Lambda_n(\Omega_M),\rho^{\otimes n}\bigr)
\leq\varepsilon
\end{array}
\right.\right\},
\label{eq:fixed-error-cost-block}
\end{equation}
where $\Omega_M$ is a maximally entangled state of Schmidt rank $M$.
The corresponding fixed-error asymptotic rate is
\begin{equation}
C_{\mathrm{op}}^\varepsilon(\rho)
=
\limsup_{n\to+\infty}\frac1n\, C_n^\varepsilon(\rho).
\label{eq:fixed-error-cost-rate}
\end{equation}
This definition keeps the error tolerance $\varepsilon$  fixed as $n$ tends to
infinity. It is therefore distinct, by definition, from the
usual entanglement cost, where the preparation error is required
to vanish asymptotically.\footnote{A closely related operational approach was developed by
Buscemi and Datta~\cite{Bus}, who studied the
one-shot entanglement cost under LOCC. In their setting, the
cost is defined as the minimum logarithm of the Schmidt rank
of a maximally entangled state from which the target state can
be prepared, with a prescribed accuracy. More precisely, the
accuracy requirement is expressed in terms of the fidelity:
the output state $\omega$ must satisfy
$F^2(\omega,\rho)\geq 1-\varepsilon$.
They derived upper and lower bounds on this one-shot cost in
terms of smoothed conditional zero-R\'enyi entropies and used
these bounds to recover the standard identity
$E_C(\rho)=E_F^\infty(\rho)$ in the asymptotic limit.

The operational quantity $C_{\mathrm{op}}^\varepsilon(\rho)$
defined above is different in two respects. First, we formulate
the accuracy constraint using the trace distance rather than
the fidelity. Second, we keep the error tolerance $\varepsilon$
fixed as $n\to+\infty$, instead of subsequently taking the limit
$\varepsilon\to0$ used in the standard asymptotic entanglement
cost.}

The basic relation between $E^\infty_{F,\varepsilon}$ and this operational
quantity is immediate from LOCC monotonicity of the EoF.
Suppose an LOCC protocol using the state $\Omega_M$ produces a state
$\omega_n$ such that
\[
T(\omega_n,\rho^{\otimes n})\leq\varepsilon.
\]
Then
\[
E_{F,\varepsilon}(\rho^{\otimes n})
\leq E_F(\omega_n)
\leq E_F(\Omega_M)
=\ln M.
\]
Taking the infimum over admissible protocols gives
\[
E_{F,\varepsilon}(\rho^{\otimes n})
\leq C_n^\varepsilon(\rho).
\]
After division by $n$ and passage to the upper limit, we obtain
\begin{equation}
E^\infty_{F,\varepsilon}(\rho)
\leq C_{\mathrm{op}}^\varepsilon(\rho).
\label{eq:R-operational-bound}
\end{equation}

For maximally entangled states this bound is strict whenever
$0<\varepsilon<1$. We will prove below that
\begin{equation}
C_{\mathrm{op}}^\varepsilon(\Omega_d)=\ln d,
\qquad 0\leq\varepsilon<1.
\label{eq:operational-Phi}
\end{equation}
while \eqref{eq:R-Phi-exact} implies
\[
E^\infty_{F,\varepsilon}(\Omega_d)
=
(1-\varepsilon)\ln d
<
\ln d
=
C_{\mathrm{op}}^\varepsilon(\Omega_d)
\]
for $0<\varepsilon<1$. This example shows that the
regularization of the smoothed EoF may be not equal to the
operational preparation cost at fixed error.\medskip

To prove (\ref{eq:operational-Phi}), put $D=d^n$, so that,
under the canonical identification of the tensor factors,
\[
  \Omega_d^{\otimes n}=\Omega_D,
  \qquad
  |\Omega_D\rangle=\frac{1}{\sqrt D}
  \sum_{j=0}^{D-1}|j\rangle|j\rangle.
\]

To prove that $"\geq"$ holds in (\ref{eq:operational-Phi}) suppose that a deterministic LOCC
protocol using a maximally entangled state $\Omega_M$ produces a
state $\omega_n$ satisfying
\[
  T(\omega_n,\Omega_D)\leq\varepsilon.
\]
Since a LOCC operation cannot increase
the Schmidt number of a state, the Schmidt number of $\omega_n$
is at most $M$.  We next show that every bipartite state $\omega$ of Schmidt number
at most $M$ satisfies
\begin{equation}
  \langle\Omega_D|\omega|\Omega_D\rangle
  \leq \frac{M}{D}.
  \label{eq:schmidt-number-overlap}
\end{equation}
To this end, consider first a  pure state $|\psi\rangle$
of Schmidt rank $r\leq M$, with the Schmidt decomposition
\[
  |\psi\rangle
  =\sum_{i=1}^{r}\sqrt{p_i}\,
    |\alpha_i\rangle|\beta_i\rangle,
  \qquad
  p_i\geq0,\quad \sum_{i=1}^{r}p_i=1,\quad \|\alpha_i\|=\|\beta_i\|=1.
\]
Then
\[
  \langle\Omega_D|\psi\rangle
  =\frac{1}{\sqrt D}
    \sum_{i=1}^{r}\sqrt{p_i}
    \sum_{j=0}^{D-1}
       \langle j|\alpha_i\rangle\langle j|\beta_i\rangle.
\]
For every $i$, the Cauchy--Schwarz inequality gives
\[
  \left|
    \sum_{j=0}^{D-1}
       \langle j|\alpha_i\rangle\langle j|\beta_i\rangle
  \right|
  \leq 1,
\]
because $|\alpha_i\rangle$ and $|\beta_i\rangle$ are unit vectors. Therefore,
using Cauchy--Schwarz inequality once more, we obtain
\[
  |\langle\Omega_D|\psi\rangle|
  \leq\frac{1}{\sqrt D}\sum_{i=1}^{r}\sqrt{p_i}\leq \sqrt{\frac{r}{D}}
  \leq \sqrt{\frac{M}{D}}.
\]
Consequently,
\[
\langle\Omega_D|\psi\rangle\langle\psi|\Omega_D\rangle=|\langle\Omega_D|\psi\rangle|^2\leq\frac{M}{D}.
\]
Since any mixed state $\,\omega\,$ of Schmidt number at most $M$ admits a convex
decomposition into pure states of Schmidt rank at most $M$, this proves
\eqref{eq:schmidt-number-overlap}.

Now, since the spectral diameter of the operator $\Omega_D$ is equal to $1$, we have
\[
  1-\langle\Omega_D|\omega_n|\Omega_D\rangle
  =\left|\operatorname{Tr}
       \Omega_D(\Omega_D-\omega_n)\right|\leq T(\Omega_D,\omega_n)
  \leq\varepsilon.
\]
It follows that
\[
  1-\varepsilon
  \leq \langle\Omega_D|\omega_n|\Omega_D\rangle
  \leq \frac{M}{D}.
\]
Hence every admissible protocol must satisfy
\begin{equation*}
  M\geq (1-\varepsilon)D
    =(1-\varepsilon)d^n.
\end{equation*}
Taking logarithms and then the infimum over all admissible
protocols, we conclude that
\begin{equation}
  C_n^\varepsilon(\Omega_d)
  \geq n\ln d+\ln(1-\varepsilon),
  \qquad 0\leq\varepsilon<1.
  \label{eq:operational-cost-lower-bound}
\end{equation}

To prove that $"\leq"$ holds in (\ref{eq:operational-Phi}), observe that the state
$\Omega_d^{\otimes n}$ is itself maximally entangled of Schmidt
rank $d^n$. It can therefore be prepared exactly by the identity
LOCC channel from the resource state $\Omega_{d^n}$. In particular,
this protocol satisfies the error constraint for every
$\varepsilon\geq0$ and gives
\begin{equation}
  C_n^\varepsilon(\Omega_d)\leq\ln(d^n)=n\ln d.
  \label{eq:operational-cost-upper-bound}
\end{equation}

Combining \eqref{eq:operational-cost-lower-bound} and
\eqref{eq:operational-cost-upper-bound}, we obtain
\[
C_{\mathrm{op}}^\varepsilon(\Omega_d)=\lim_{n\to\infty}\frac{1}{n}\,C_n^\varepsilon(\Omega_d)
  =\ln d, \qquad 0\leq\varepsilon<1.
\]

\section{Summary of the results}

In this article, we have explored the $\varepsilon$-smoothed Entanglement of Formation (EoF), defined for a bipartite state $\rho$ and $\varepsilon\in[0,1]$ by
\[
 E_{F,\varepsilon}(\rho)
 =\inf_{\sigma:\,T(\rho,\sigma)\leq\varepsilon}E_F(\sigma),
 \qquad T(\rho,\sigma)=\tfrac12\|\rho-\sigma\|_1.
\]
This construction quantifies the entanglement that is guaranteed when the state is specified only up to a given trace-distance error. In particular, $E_{F,0}=E_F$ and $E_{F,1}=0$.

A central result of the article is a useful relation between the $\varepsilon$-smoothed EoF and the Moreau--Yosida approximation $E_F^\lambda$ of the EoF explored in \cite{EF}. It is proved that for every bipartite state $\rho$,
\[
 E_F^\lambda(\rho)=\min_{0\leq\varepsilon\leq1}
 \left\{E_{F,\varepsilon}(\rho)+\frac{\varepsilon}{\lambda}\right\},\;\qquad  \lambda>0
\]
(where the minimum is always attained) and that
\[
 E_{F,\varepsilon}(\rho)=\sup_{\lambda>0}
 \left\{E_F^\lambda(\rho)-\frac{\varepsilon}{\lambda}\right\},\qquad \varepsilon\in(0,1].
\]
The supremum in the latter representation need not be attained, even for two-qubit states. These dual formulas allow us to transfer properties of the Moreau--Yosida approximation to the smoothed EoF and to obtain a variational expression for $E_{F,\varepsilon}$.\smallskip

We have established the following basic properties of the smoothed EoF:
\begin{itemize}
\item For every $\varepsilon>0$, the value $E_{F,\varepsilon}(\rho)$ is finite, including in infinite-dimensional systems. For each fixed state $\rho$, the function $\,\varepsilon\mapsto E_{F,\varepsilon}(\rho)\in[0,+\infty]\,$ is continuous, nonincreasing and convex on $[0,1]$, and
\[
 E_{F,\varepsilon}(\rho)\nearrow E_F(\rho)\leq+\infty
 \qquad\text{as }\quad\varepsilon\searrow0.
\]
If $\,E_F(\rho)<+\infty$, then $\,E_{F,\varepsilon}(\rho)\leq(1-\varepsilon)E_F(\rho)$.
\item The equality $\,E_{F,\varepsilon}(\rho)=0\,$ holds if and \emph{only if} the closed trace-distance $\varepsilon$-vicinity of $\rho$ contains a separable state. Writing $d\!\left(\rho,\St^{AB}_{\rm sep}\right)$ for the trace distance from $\rho$ to the set $\St^{AB}_{\rm sep}$ of separable states, we obtain, whenever $d\!\left(\rho,\St^{AB}_{\rm sep}\right)>\varepsilon$,
\[
 E_{F,\varepsilon}(\rho)\geq
 2\bigl(d\!\left(\rho,\St^{AB}_{\rm sep}\right)-\varepsilon\bigr)^2>0.
\]

\item For every fixed $\varepsilon>0$, the function $\rho\mapsto E_{F,\varepsilon}(\rho)$ is convex and continuous, including in infinite-dimensional systems. It is Lipschitz continuous on sets of states, on which $E_F$ is bounded, and does not increase under nonselective LOCC operations.
\item We obtain the  variational expression
\[
 E_{F,\varepsilon}(\rho)=\sup_{\Lambda\in \mathfrak{F}(\varepsilon)}
 \operatorname{Tr}\Lambda\rho,
\]
where, writing $D(\Lambda)$ for the spectral diameter of $\Lambda$,
\begin{align*}
\mathfrak{F}(\varepsilon)&\doteq\left\{\Lambda \in\B_{\rm sa}(\HH_{AB})\left|\, \exists\lambda>0\,:\,D(\Lambda)\leq \frac{1}{\lambda},\,\langle\psi|\Lambda|\psi\rangle
     \leq S(\psi_A)-\frac{\varepsilon}{\lambda}\;\; \forall\psi\in\HH^1_{AB}\!\right.\right\}\\
&=\left\{\Lambda\in\mathfrak B_{\rm sa}(\mathcal H_{AB})\left|\,\exists\lambda>0\, : \,D(\Lambda)\leq\frac1\lambda,\, \operatorname{Tr}\Lambda\sigma\leq E_F(\sigma)-\frac{\varepsilon}{\lambda}\;\; \forall\sigma\in\mathfrak S(\mathcal H_{AB})\!\right.\right\}
\end{align*}
($\HH^1_{AB}$ is the unit sphere in $\HH_{AB}$).  At $\varepsilon=0$, this representation reduces to the standard variational expression for the EoF.
\item The smoothed EoF satisfies the weak subadditivity inequality
\[
 E_{F,\varepsilon_1+\varepsilon_2}(\rho\otimes\sigma)
 \leq E_{F,\varepsilon_1}(\rho)+E_{F,\varepsilon_2}(\sigma),
 \qquad \varepsilon_1,\varepsilon_2\geq 0,\;\;\;   \varepsilon_1+\varepsilon_2\leq1.
\]
In contrast, subadditivity at a fixed smoothing radius fails in general. We demonstrate this already for two-qubit maximally entangled states.
\item Explicit estimates on the approximation error $\,E_F(\rho)-E_{F,\varepsilon}(\rho)\,$ applicable to any bipartite state $\rho$ with finite marginal rank/energy are obtained. These estimates give bounds on the rate of uniform convergence of $E_{F,\varepsilon}$ to $E_F$ on sets of states with bounded rank/energy of one of the marginal states.
\end{itemize}
\medskip

We have also compared the quantities
\[
 E_{F,\varepsilon}^{+\infty}(\rho)
 \doteq\limsup_{n\to+\infty}\frac1nE_{F,\varepsilon}(\rho^{\otimes n}),
 \qquad
 [E_F^\infty]_\varepsilon(\rho)
 \doteq\inf_{\sigma:\,T(\rho,\sigma)\leq\varepsilon}E_F^\infty(\sigma),
\]
and showed that they  do not coincide in general. Whether $E_{F,\varepsilon}^{\infty}(\rho)\geq [E_F^\infty]_\varepsilon(\rho)$ holds for every bipartite state $\rho$ remains open. We further show that $E_{F,\varepsilon}^{\infty}(\rho)$ is bounded above by the operational entanglement cost at fixed error, and that this bound is strict for maximally entangled states whenever $\,0<\varepsilon<1$.

\section{Possible generalizations}
\label{sec:possible-generalizations}

The  Mora-Piani-Briegel construction of  $\varepsilon$-measure (proposed in \cite{MPB}) can be applied to an arbitrary
entanglement measure on the set of states of a bipartite quantum system $AB$ \cite{P&V,4H} , i.e. such a convex function
$E_X$ on the set $\,\St(\mathcal H_{AB})\,$ of states of this system taking values in $\,[0,+\infty]\,$  and satisfying the following conditions:
\begin{enumerate}[1)]
\item $\;E_X(\rho)=0\quad\Longleftrightarrow\quad
 \rho\,\textrm{ is a separable state}$;
\item $\;E_X(\rho)=S(\rho_A)=S(\rho_B)\,$ for every pure state $\,\rho\in\mathcal \St(\mathcal H_{AB})$;
\item selective LOCC monotonicity.
\end{enumerate}

Define the $\varepsilon$-smoothed version of an entanglement measure $E_X$ as
\begin{equation}
E_{X,\varepsilon}(\rho)\doteq
\inf_{\sigma:\,T(\rho,\sigma)\leq\varepsilon}E_X(\sigma),
\qquad 0\leq\varepsilon\leq1,
\label{eq:general-MPB}
\end{equation}
and introduce the associated Moreau--Yosida approximation
\begin{equation}
 E_X^\lambda(\rho)\doteq
 \inf_{\sigma\in\St(\mathcal H_{AB})}
 \left\{E_X(\sigma)+\frac{1}{\lambda}T(\rho,\sigma)\right\},
 \qquad \lambda>0.
 \label{eq:general-MY}
\end{equation}

A detailed analysis of the proof of the results in the main part of the article provided with some general facts from convex analysis
allow us to prove the following\smallskip

\begin{proposition}\label{gen}  \emph{Let $E_X$ be a convex lower-semicontinous function
on the set $\,\St(\HH_{AB})$  taking values in $\,[0,+\infty]\,$ and  satisfying the above conditions 1-3.
Then all the results on properties of the functions $E_{F,\varepsilon}$ and  $E_F^\lambda$  presented
in Sections 3.1-3.3 are valid for the functions $E_{X,\varepsilon}$ and  $E_X^\lambda$ defined in (\ref{eq:general-MPB}) and  (\ref{eq:general-MY})
except for the lower bound (\ref{eq:EF-lower-bound}) and the statements involving the particular ensemble
representation of the EoF and the corresponding set of measures $\P_{\rho,\varepsilon}(\HH_{AB})$ and its subsets. The variational expression for
$E_{X,\varepsilon}$ has the form
\[
 E_{X,\varepsilon}(\rho)=\sup_{\Lambda\in \mathfrak{X}(\varepsilon)}
 \operatorname{Tr}\Lambda\rho,
\]
where, writing $D(\Lambda)$ for the spectral diameter of $\Lambda$,}
\begin{equation*}
\mathfrak{X}(\varepsilon)=\left\{\Lambda\in\mathfrak B_{\rm sa}(\mathcal H_{AB})\left|\,\exists\lambda>0\, : \,D(\Lambda)\leq\frac1\lambda,\, \operatorname{Tr}\Lambda\sigma\leq E_X(\sigma)-\frac{\varepsilon}{\lambda}\;\; \forall\sigma\in\mathfrak S(\mathcal H_{AB})\!\right.\right\}.
\end{equation*}

\emph{In addition,
\begin{itemize}
  \item if $\,E_X(\rho)\geq E_R(\rho)\,$ for all $\,\rho\in \St(\HH_{AB}),$ then the lower bound (\ref{eq:EF-lower-bound}) is valid for the function $E_{X,\varepsilon}$;
  \item if the function $E_X$ is subadditive w.r.t. tensor products, then all the results presented
in Sections 3.4 are valid  for the function $E_{X,\varepsilon}$.
\end{itemize}}
\end{proposition} \smallskip

\begin{remark}\label{q-a}
The qualitative assertion
\begin{equation}\label{qqq}
E_{X,\varepsilon}(\rho)=0
\quad\Longleftrightarrow\quad
U_\varepsilon(\rho)\cap
\St_{\mathrm{sep}}^{AB}\ne\varnothing,
\end{equation}
where $\,\St_{\mathrm{sep}}^{AB}$ is the set of separable states and
\[
 U_\varepsilon(\rho)\doteq
 \{\sigma\in\St(\mathcal H_{AB})\,|\,
 T(\rho,\sigma)\le\varepsilon\},
\]
holds without the additional condition $\,E_X(\rho)\ge E_R(\rho),$ which is needed
for the quantitative lower bound in (\ref{eq:EF-lower-bound}), but not for
this qualitative equivalence.

To prove (\ref{qqq}) without using the condition $\,E_X(\rho)\ge E_R(\rho)\,$  assume that
$E_{X,\varepsilon}(\rho)=0$ for some $\varepsilon>0$ (the case $\varepsilon=0$ is trivial). Then there exists a sequence of states $\{\sigma_n\}$ such that
\begin{equation}
 T(\rho,\sigma_n)\le\varepsilon,
 \qquad
 \lim_{n\to+\infty}E_X(\sigma_n)=0.
 \label{eq:general-zero-sequence}
\end{equation}

Let $\{\Phi_k\}$ be a sequence of local truncation channels with
finite-dimensional output spaces such that
\begin{equation}\label{scc}
\lim_{k\to\infty}T(\Phi_k(\omega),\omega)=0
\end{equation}
for every state $\omega\in\St(\HH_{AB})$ (for example, the channels defined in (\ref{ch-form}) with pure product state $\tau$ and increasing sequences of projectors $P^k_A$ and $P^k_B$ strongly converging to the unit operators $I_A$ and $I_B$).

Fix $k$. Since the channel $\Phi_k$ is trace-distance contractive and the function $E_X$
is non-increasing under LOCC, we have
\[
 T\bigl(\Phi_k(\rho),\Phi_k(\sigma_n)\bigr)
 \le\varepsilon,
 \qquad
 E_X\bigl(\Phi_k(\sigma_n)\bigr)
 \le E_X(\sigma_n).
\]
The set of states supported on the finite-dimensional output space
of $\Phi_k$ is compact. Hence, after passing to a subsequence
(which may depend on $k$), we may assume that
\[
 \Phi_k(\sigma_{n_j})\rightarrow\omega_k\quad\textrm{ as }\quad j\to +\infty
\]
in trace norm for some state $\omega_k$ in that finite-dimensional
space (the support of the range of $\Phi_k$). By lower semicontinuity of $E_X$ and
\eqref{eq:general-zero-sequence} we have
\[
 0\le E_X(\omega_k)
 \le\liminf_{j\to+\infty}E_X\bigl(\Phi_k(\sigma_{n_j})\bigr)
 =0.
\]
Thus $E_X(\omega_k)=0$. By the assumed validity of the above property 1 for the function $E_X$, the state $\omega_k$ is separable. Moreover,
continuity of the trace distance gives $\,T\bigl(\Phi_k(\rho),\omega_k\bigr)\le\varepsilon$.

Consequently, by the triangle inequality we have
\[
 d\!\left(\rho,\St^{AB}_{\rm sep}\right)
\doteq\inf_{\tau\in\St_{\mathrm{sep}}^{AB}}
           T(\rho,\tau)\le T(\rho,\omega_k)\le T\bigl(\rho,\Phi_k(\rho)\bigr)+\varepsilon.
\]
By taking the limit as $\,k\to+\infty$ and using (\ref{scc}), we obtain
$\,d\!\left(\rho,\St^{AB}_{\rm sep}\right)\le\varepsilon$.

By Lemma \ref{atd} in the Appendix, the distance from any bipartite
state to the set $\,\St^{AB}_{\rm sep}\,$ is attained. Hence there
exists a separable state $\tau_\rho$ such that
\[
 T(\rho,\tau_\rho)= d\!\left(\rho,\St^{AB}_{\rm sep}\right)\le\varepsilon.
\]
It follows that
$\tau_\rho\in U_\varepsilon(\rho)\cap
\St^{AB}_{\rm sep}$,  proving the
nontrivial implication $"\Rightarrow"$ in \eqref{qqq}. $\Box$
\end{remark}
\medskip

The lower semicontinuity of the function $E_X$ is used in many places. In particular, it is necessary
to prove that
\begin{equation}\label{LSC}
E_{X,\varepsilon}(\rho)\nearrow E_{X,0}(\rho)\doteq E_X(\rho)\leq +\infty \quad\textup{ as }\quad \varepsilon\searrow 0^+.
\end{equation}
Indeed, since the continuity of the function $E_{X,\varepsilon}$ (with $\varepsilon>0$) on $\St(\HH_{AB})$ can be proved without using the lower semicontinuity property of $E_X$, it is not hard to show that
(\ref{LSC}) is equivalent to local lower semicontinuity of the function $E_X$ at the state $\rho$.

The lower semicontinuity of the function $E_X$  is essentially used in the proof of the second of the variational relations
$$
E_X^\lambda(\rho)
=
\inf_{0\leq\varepsilon\leq1}
\left\{
E_{X,\varepsilon}(\rho)+\frac{\varepsilon}{\lambda}
\right\}
\quad
\textrm{and}
\quad
E_{X,\varepsilon}(\rho)
=
\sup_{\lambda>0}
\left\{
E_X^\lambda(\rho)-\frac{\varepsilon}{\lambda}
\right\}
$$
which many of the results of this article are based on.

At the same time, such analytical properties as convexity, global continuity, weak subadditivity and nonselective LOCC monotonicity
of the function $E_{X,\varepsilon}$ as well as the first of the above variational relations are proved  without the lower semicontinuity assumption.\medskip

By applying Proposition \ref{gen}  to the relative entropy of
entanglement
\[
E_R(\rho)
=
\inf_{\omega\in\St_{\rm sep}^{AB}}
D(\rho\shs\Vert\omega),\quad \rho\in \St(\HH_{AB}),
\]
where $\,D(\cdot\Vert\cdot)\,$ is the quantum relative entropy \cite{P&V,4H,L-Sh-1}, we obtain the following.\smallskip

\begin{corollary}\label{gen++}  \emph{All the results on properties of the functions $E_{F,\varepsilon}$ and  $E_F^\lambda$  presented
in Sections 3.1-3.4 are valid for the functions $E_{R,\varepsilon}$ and  $E_R^\lambda$ defined in (\ref{eq:general-MPB}) and  (\ref{eq:general-MY})
with $E_X=E_R$ except for statements involving the particular ensemble
representation of the EoF and the corresponding set of measures $\P_{\rho,\varepsilon}(\HH_{AB})$ and its subsets.}
\end{corollary} \smallskip

\emph{Proof.} It suffices to note that $E_R$  is a convex lower-semicontinous subadditive function
on the set $\St(\HH_{AB})$  taking values in $\,[0,+\infty]\,$ and  satisfying the above conditions 1-3 \cite{P&V,4H,L-Sh-1}. $\Box$

It is worth to note that the global lower-semicontinuity of $E_R$ in the infinite-dimensional settings is proved in \cite{L-Sh-1}
using the same topological result (the weak*-closedness of the cone generated by the set  $\St_{\rm sep}(\mathcal H_{AB})$) which is used
in proving the existence of a nearest separable state for any state of a bipartite infinite-dimensional quantum system $AB$ presented in the Appendix.\smallskip

In finite-dimensions, explicit estimates on the approximation error $\,E_R(\rho)-E_{R,\varepsilon}(\rho)\,$  can be obtained (via the $E_R$-analog of Lemma \ref{vil})
by applying any continuity bound for the function $E_R$ \cite{D&H,W-CB,L-Sh-1}.  In infinite-dimensions, to obtain such estimates we need  $E_R$-analogs of the semicontinuity bounds for
the function $E_F$ with the rank/energy constraint presented in \cite{LCB+,AOE}  (as far as I know, no such analogues have been obtained yet).\newpage

\section{Appendix: On existence of a nearest separable state}

The set $\St^{AB}_{\rm sep}$ of separable states of a bipartite quantum system $AB$ is defined as a trace norm convex closure of the set of
all product states in $\St(\mathcal H_{AB})$  \cite{H-SCI,Wilde}. The distance from an arbitrary state $\rho$ in $\St(\mathcal H_{AB})$ to the set $\St^{AB}_{\rm sep}$ w.r.t. the trace norm metric  is defined as
\[
 d\!\left(\rho,\St^{AB}_{\rm sep}\right)=\textstyle\frac{1}{2}\inf_{\sigma\in\St^{AB}_{\rm sep}}\|\rho-\sigma\|_1.
\]
If at least one of the spaces $\HH_A$ and  $\HH_B$ is infinite-dimensional, then
the closed set $\St^{AB}_{\rm sep}$ is not compact. So, special arguments are required to show that for any state $\rho$ in $\St(\mathcal H_{AB})$
there exists a nearest separable state, i.e. such a state $\omega_\rho\in \St_{\rm sep}^{AB}$ at which the above
infimum is attained. Fortunately, such arguments can be found by using the weak*-closedness\footnote{The weak-* topology on $\T(\HH_{AB})$ is defined in the below proof of Lemma \ref{atd}.} of the cone generated by the set  $\St^{AB}_{\rm sep}$ \cite{L-Sh-1}.
\smallskip

\begin{lemma}\label{atd}
\emph{Let $A$ and $B$ be quantum systems of any dimensions. For every bipartite state $\rho$ of the system $AB$, there exists a state $\,\omega_\rho\in\St_{\rm sep}(\mathcal H_{AB})\,$ such that}
\begin{equation}\label{atd+}
 d\!\left(\rho,\St^{AB}_{\rm sep}\right)=\textstyle\frac{1}{2}\|\rho-\omega_\rho\|_1.
\end{equation}
\emph{Moreover, if $\,\{\omega_n\}\subset\St_{\rm sep}^{AB}\,$ is a minimizing sequence\footnote{This means that $\,\textstyle\frac{1}{2}\|\rho-\omega_n\|_1\to d\!\left(\rho,\St^{AB}_{\rm sep}\right)\,$ as $\,n\to+\infty$.} for a state $\rho$ then
\begin{itemize}
  \item the set $\,\mathfrak{L}(\{\omega_n\})\,$ of  weak-*  partial limits of the sequence $\,\{\omega_n\}\,$ is nonempty;
  \item equality (\ref{atd+}) holds with  $\,\omega_\rho=\tau+(1-\operatorname{Tr}\tau)\omega_0,$ where $\tau$ is  any operator in $\,\mathfrak{L}(\{\omega_n\})\,$ and $\omega_0$ is any state in $\,\St^{AB}_{\rm sep}$.
\end{itemize}}
\end{lemma}
\medskip
\begin{remark}\label{ard-r}
Example \ref{nue} below shows that a nearest separable state is not unique, in general.
\end{remark}\smallskip

\emph{Proof.} Choose a minimizing sequence $\,\{\omega_n\}\subset\St^{AB}_{\rm sep}\,$ such that
\[
 \|\rho-\omega_n\|_1\rightarrow 2d(\rho,\St^{AB}_{\rm sep})\quad\textrm{ as }\;\;n\to +\infty.
\]
This sequence is bounded in trace norm. Since the space $\mathcal H_{AB}$ is separable, the unit ball of the trace-class operators is metrizable in the weak-* topology on $\T(\HH_{AB})$ induced
by the Banach space of all compact operators on $\mathcal H_{AB}$ with the operator norm  (see details in \cite[Section 2.1]{L-Sh-1}). Hence, the Banach--Alaoglu theorem implies the relative compactness of $\St^{AB}_{\rm sep}$ in the weak-* topology. So, the above sequence contains a convergent subsequence.  Relabel it as $\{\omega_n\}$ and assume that $\tau$ is its weak-* limit.

Since the unit ball in $\T(\HH_{AB})$ is closed in the weak-* topology, we conclude that  $\operatorname{Tr}\tau\leq1$. Moreover,  $\tau\in\operatorname{cone}(\St^{AB}_{\rm sep})$ because of the weak*-closedness of this cone \cite[Corollary 22]{L-Sh-1}. Write $r=\operatorname{Tr}\tau$.

We next prove the key estimate
\begin{equation}
 2d(\rho,\St^{AB}_{\rm sep})
 \geq \|\rho-\tau\|_1+(1-r).
 \label{eq:key-estimate}
\end{equation}
Let $P$ be any finite-rank orthogonal projector and put $Q=I-P$. For each $n$, the operator $\delta_n=\rho-\omega_n$ is self-adjoint. The pinching inequality for the trace norm gives
\[
 \|\delta_n\|_1\geq\|P\delta_nP\|_1+\|Q\delta_nQ\|_1.
\]
Since $|\operatorname{Tr}\varsigma|\leq\|\varsigma\|_1$ for every self-adjoint trace-class operator $\varsigma$, it follows that
\begin{align}
 \|\rho-\omega_n\|_1
 &\geq \|P(\rho-\omega_n)P\|_1
       +\big|\operatorname{Tr}\big(Q(\rho-\omega_n)\big)\big|. \label{eq:pinching}
\end{align}
For a fixed finite-rank operator $P$, the weak-* convergence implies the convergence of the  operators $P(\rho-\omega_n)P$ in trace norm and convergence of their traces. So, since $\operatorname{Tr}\rho=\operatorname{Tr}\omega_n=1$, we have
\[
 \operatorname{Tr}Q(\rho-\omega_n)
 =\operatorname{Tr}P\omega_n-\operatorname{Tr}P\rho\;\;\rightarrow\;\;
 \operatorname{Tr}P\tau-\operatorname{Tr}P\rho\quad\textrm{ as }\;\;n\to +\infty.
\]
Thus, by taking the limit in \eqref{eq:pinching} as $\,n\to+\infty\,$ we obtain
\begin{align}
 2d(\rho,\St^{AB}_{\rm sep})
 &\geq \|P(\rho-\tau)P\|_1
       +\big|\!\operatorname{Tr}P\tau-\operatorname{Tr}P\rho\,\big|.
 \label{eq:finiteP}
\end{align}
Now choose an increasing sequence of finite-rank projectors $P_k$ strongly converging to the unit operator $I_{AB}$ . For every trace-class operator $\varsigma$,
$\,\|P_k\varsigma P_k-\varsigma\|_1\rightarrow0\,$ as $\,k\to +\infty$. In particular,
\[
 \|P_k(\rho-\tau)P_k\|_1\rightarrow\|\rho-\tau\|_1,\quad
 \operatorname{Tr}P_k\rho\rightarrow1,
 \qquad
 \operatorname{Tr}P_k\tau\rightarrow r\quad\textrm{ as }\;\;k\to +\infty.
\]
Taking $\,k\to+\infty\,$ in \eqref{eq:finiteP} with $P$ replaced by $P_k$ proves \eqref{eq:key-estimate}.\medskip

Fix any separable state $\omega_0$. Since $\tau\in\operatorname{cone}(\St^{AB}_{\rm sep})$ and $\operatorname{Tr}\tau=r$, the operator
\[
 \widetilde\omega\doteq\tau+(1-r)\omega_0
\]
is a state in $\St^{AB}_{\rm sep}$. By the definition of $d(\rho,\St^{AB}_{\rm sep})$  and \eqref{eq:key-estimate} we have
\[
 2d(\rho,\St^{AB}_{\rm sep})
 \leq\|\rho-\widetilde\omega\|_1
 \leq\|\rho-\tau\|_1+(1-r)
 \leq 2d(\rho,\St^{AB}_{\rm sep}).
\]
So all the inequalities are equalities and, hence, $\widetilde\omega$ is a closest separable state to the state $\rho$. $\Box$ \bigskip

The following example shows that the nearest separable state in Lemma \ref{atd} need not be unique,
even for a two-qubit maximally entangled state.\smallskip\newpage

\begin{example}\label{nue} Let $\,A=B=\mathbb{C}^2.\,$ Consider the maximally entangled state
\[
  \Omega_2=|\phi_+\rangle\langle\phi_+|,
  \qquad
  |\phi_+\rangle=\frac{|00\rangle+|11\rangle}{\sqrt{2}}.
\]
Let
\[
  |\phi_-\rangle=\frac{|00\rangle-|11\rangle}{\sqrt{2}},
  \qquad
  |\psi_+\rangle=\frac{|01\rangle+|10\rangle}{\sqrt{2}},
\]
and consider the  separable states
\[
  \tau_z=\frac12\left(
       |\phi_+\rangle\langle\phi_+|
      +|\phi_-\rangle\langle\phi_-|
  \right)=\frac12\bigl(|00\rangle\langle00|
                      +|11\rangle\langle11|\bigr)
\]
and
\[
  \tau_x=\frac12\left(
       |\phi_+\rangle\langle\phi_+|
      +|\psi_+\rangle\langle\psi_+|
  \right)=\frac12\bigl(|\vartheta_+\otimes\vartheta_+\rangle\langle\vartheta_+\otimes\vartheta_+|
                      +|\vartheta_-\otimes\vartheta_-\rangle\langle\vartheta_-\otimes\vartheta_-|\bigr),
\]
where $\,|\vartheta_\pm\rangle=\frac{|0\rangle\pm|1\rangle}{\sqrt{2}}$. Both states $\tau_z$ and $\tau_x$ are separable, since they are convex
combinations of pure product states. Moreover, they are distinct.

We have
\[
\begin{aligned}
  \Omega_2-\tau_z
    &=\frac12\left(
       |\phi_+\rangle\langle\phi_+|
      -|\phi_-\rangle\langle\phi_-|
    \right),\\
  \Omega_2-\tau_x
    &=\frac12\left(
       |\phi_+\rangle\langle\phi_+|
      -|\psi_+\rangle\langle\psi_+|
    \right).
\end{aligned}
\]
Since the vectors in each difference are orthogonal and normalized,
both differences have trace norm equal to $1$. Hence, $\,T(\Omega_2,\tau_z)=T(\Omega_2,\tau_x)=\frac12$.

It remains to show that this distance is minimal. For any
unit vectors $\,|\alpha\rangle=(a_0,a_1)^\top$ and $\,|\beta\rangle=(b_0,b_1)^\top$ in $\mathbb C^2$,
the Cauchy--Schwarz inequality gives
\[
  |\langle\phi_+|\alpha\otimes\beta\rangle|^2
  =\frac12|a_0b_0+a_1b_1|^2
  \leq\frac12.
\]
By convexity, every separable state $\omega$ therefore satisfies $\,\operatorname{Tr}\Omega_2\omega\leq\frac12$.
Since the spectral diameter of the operator $\Omega_2$ is equal to $1$, we have
\[
  \frac12\|\Omega_2-\omega\|_1
  \geq
  \left|\operatorname{Tr}\Omega_2(\Omega_2-\omega)\right|=1-\operatorname{Tr}\Omega_2\omega\geq\frac12.
\]
Thus $\,\frac12\inf_{\omega\in \St^{AB}_{\rm sep}}\|\Omega_2-\omega\|_1=\frac12$,
and both $\tau_z$ and $\tau_x$ attain this infimum. Since
$\tau_z\ne\tau_x$, the nearest separable state is not unique.
\end{example}

\paragraph{Use of AI tools.}
In preparing this work,  ChatGPT 5.6 (free version) was used. The AI agent assisted with literature search,
finding counter-examples and optimizing proofs. All mathematical content was checked by the author, who takes full
responsibility for it.
\medskip

I am grateful to L.Lami for the recent collaboration that helped me in  this research.

\end{document}